\documentclass[11pt,letterpaper]{amsart}
\usepackage{amsmath,amsfonts,amssymb,amsthm,mathscinet}
\usepackage{enumerate}
\usepackage{color,xcolor}
\usepackage{colortbl,multirow}
\usepackage{cite, hyperref}
\usepackage[noabbrev,capitalize]{cleveref}

\DeclareMathOperator{\RePart}{Re}
\DeclareMathOperator{\PC}{PC}
\DeclareMathOperator{\PCC}{PCC}
\DeclareMathOperator{\PSL}{PSL}
\DeclareMathOperator{\sig}{sig}
\renewcommand{\Re}{\RePart}
\newcommand{\C}{{\mathbb C}}
\newcommand{\N}{{\mathbb N}}
\newcommand{\Q}{{\mathbb Q}}
\newcommand{\R}{{\mathbb R}}
\newcommand{\Z}{{\mathbb Z}}
\newcommand{\conj}[1]{\overline{#1}}
\newcommand{\floor}[1]{\lfloor{#1}\rfloor}

\definecolor{tableshade}{rgb}{0.85,0.85,0.85}

\newtheorem{theorem}{Theorem}[section]
\newtheorem{proposition}[theorem]{Proposition}
\newtheorem{lemma}[theorem]{Lemma}
\newtheorem{corollary}[theorem]{Corollary}
\theoremstyle{definition}
\newtheorem{definition}[theorem]{Definition}
\newtheorem{notation}[theorem]{Notation}
\newtheorem{remark}[theorem]{Remark}

\title[Correlation of Golay--Rudin--Shapiro Sequences]{Peak Sidelobe Level and Peak Crosscorrelation of Golay--Rudin--Shapiro Sequences}
\author{Daniel J.~Katz}
\address{Department of Mathematics, California State University, Northridge, \: United States}
\author{Courtney M.~van der Linden}
\thanks{This paper is based upon work supported in part by the National Science Foundation under Grants DMS-1500856 and  CCF-1815487.  Courtney~M.~van der Linden was supported in part by a Ramanujan Research Scholarship in Mathematics from the College of Science and Mathematics at California State University, Northridge.}
\date{13 November 2021}
\begin{document}
\begin{abstract}
Sequences with low aperiodic autocorrelation and crosscorrelation are used in communications and remote sensing. Golay and Shapiro independently devised a recursive construction that produces families of complementary pairs of binary sequences.
In the simplest case, the construction produces the Rudin--Shapiro sequences, and in general it produces what we call Golay--Rudin--Shapiro sequences.
Calculations by Littlewood show that the Rudin--Shapiro sequences have low mean square autocorrelation.
A sequence's peak sidelobe level is its largest magnitude of autocorrelation over all nonzero shifts.
H\o holdt, Jensen, and Justesen showed that there is some undetermined positive constant $A$ such that the peak sidelobe level of a Rudin--Shapiro sequence of length $2^n$ is bounded above by $A(1.842626\ldots)^n$, where $1.842626\ldots$ is the positive real root of $X^4-3 X-6$.
We show that the peak sidelobe level is bounded above by $5(1.658967\ldots)^{n-4}$, where $1.658967\ldots$ is the real root of $X^3+X^2-2 X-4$.
Any exponential bound with lower base will fail to be true for almost all $n$, and any bound with the same base but a lower constant prefactor will fail to be true for at least one $n$.
We provide a similar bound on the peak crosscorrelation (largest magnitude of crosscorrelation over all shifts) between the sequences in each Rudin--Shapiro pair.
The methods that we use generalize to all families of complementary pairs produced by the Golay--Rudin--Shapiro recursion, for which we obtain bounds on the peak sidelobe level and peak crosscorrelation with the same exponential growth rate as we obtain for the original Rudin--Shapiro sequences. 
\end{abstract}
\maketitle

\section{Introduction}
Many communications and remote sensing protocols require pseudorandom sequences \cite{Golomb,Golomb-Gong,Schroeder}.
A {\it sequence} is a tuple, $(f_0,f_1,\ldots,f_{k-1})$, of complex numbers.
Since we are considering aperiodic properties of sequences, we adopt the convention that $f_j=0$ when $j \not\in \{0,1\ldots,k-1\}$, and identify our sequence $(f_0,\ldots,f_{k-1})$ with the polynomial $f(z)=\sum_{j \in \Z} f_j z^j$, which resides in the ring $\C[z,z^{-1}]$ of Laurent polynomials with complex coefficients.
A {\it binary sequence of length $\ell$} is a tuple, $(f_0,\ldots,f_{\ell-1})$, of $\ell$ elements from $\{1,-1\}$, that is, a polynomial $f(z)=\sum_{j \in \Z} f_j z^j$ where $f_j \in \{1,-1\}$ for $0 \leq j < \ell$ and $f_j=0$ otherwise.
Thus, a nonzero binary sequence of length $\ell$ is a polynomial of degree $\ell-1$.

In multi-user communication networks, it is desirable to have sequences that do not resemble translates of each other, nor should any individual sequence resemble a translated version of itself \cite{Golomb,Golomb-Gong}.
If $f$ and $g$ are two sequences and $s \in \Z$, then the {\it aperiodic crosscorrelation of $f$ with $g$ at shift $s$} is
\begin{equation}\label{Tatiana}
C_{f,g}(s)=\sum_{j \in \Z} f_{j+s} \conj{g_j}.
\end{equation}
Notice that there are only finitely many nonzero terms in this sum because  our sequences have only finitely many nonzero entries.
The {\it autocorrelation of $f$ at shift $s$} is $C_{f,f}(s)$.
We want pairs of sequences in which each sequence has low autocorrelation at all nonzero shifts and the two sequences have low crosscorrelation with each other at all shifts.

The Rudin--Shapiro sequences are family of binary sequences of increasing length, which were originally discovered by Golay \cite[p.~469]{Golay-51} and Shapiro \cite[pp.~39--40]{Shapiro-MS}, and later rediscovered by Rudin \cite[pp.~855--856]{Rudin}. The Rudin--Shapiro sequences are known to have low mean square autocorrelation (see \cite[Theorem 2.3]{Hoholdt-Jensen-Justesen} or the earlier equivalent calculation in \cite[Problem 19]{Littlewood}), and it has been shown \cite{Hoholdt-Jensen-Justesen,Borwein-Mossinghoff,Katz-Lee-Trunov-a,Katz-Lee-Trunov-b} that relatives of the Rudin--Shapiro sequences also have low correlation.
Golay devised the recursive method that produces the Rudin--Shapiro sequences and their relatives \cite[p.~469]{Golay-51} in order to construct what are now called Golay complementary pairs.
A {\it Golay complementary pair} (or {\it Golay pair} or {\it complementary pair}) is a pair of sequences $(f,g)$ such that $C_{f,f}(s)+C_{g,g}(s)=0$ for all nonzero $s \in \Z$.
Shapiro's method is equivalent to Golay's, while Rudin's construction is a special case.
We now formally define the recursion, which is most easily done using the polynomial interpretation of sequences.
\begin{definition}[Golay--Rudin--Shapiro Recursion]\label{Stan}
Let $\ell_0$ be a positive integer and let $x_0(z)$ and $y_0(z)$ be nonzero polynomials in $\C[z]$ of degree less than $\ell_0$ such that $(x_0,y_0)$ is a Golay complementary pair with $C_{x_0,x_0}(0)=C_{y_0,y_0}(0)$.
Let $\ell_n=2^n \ell_0$ for each nonnegative integer $n$.
We recursively define an infinite sequence $(x_0,y_0),(x_1,y_1),\ldots$ of pairs of polynomials where $x_n(z)=x_{n-1}(z)+z^{\ell_{n-1}} y_{n-1}(z)$ and $y_n(z)=x_{n-1}(z)-z^{\ell_{n-1}} y_{n-1}(z)$ for each $n> 0$.
Then $(x_n,y_n)$ is the {\it $n$th Golay pair obtained from $\ell_0$ and the seed pair $(x_0,y_0)$ via the Golay--Rudin--Shapiro recursion}.
The sequences produced by this construction are called {\it Golay--Rudin--Shapiro sequences}.
When one sets $\ell_0=1$ and $x_0(z)=y_0(z)=1$ in this recursion, then $x_n(z)$ and $y_n(z)$ are respectively called the {\it $n$th Rudin--Shapiro sequence} and the {\it $n$th Rudin--Shapiro companion sequence}, and $(x_n,y_n)$ is the {\it $n$th Rudin--Shapiro pair}.
\end{definition}
Observe that $x_n$ and $y_n$ are polynomials of degree less than $\ell_n$ for each nonnegative integer $n$, and that they are of degree precisely $\ell_n-1$ if $\deg(y_0)=\ell_0-1$, so that the $n$th Rudin--Shapiro sequence and its companion are binary sequences of length $2^n$.
For binary sequences, Golay indicates in \cite[p.~469]{Golay-51} (and formally proves in \cite[pp.~84--85]{Golay-61}) that each step of his construction always produces a new complementary pair from an existing one, so that every pair $(x_n,y_n)$ produced by this construction is a complementary pair.
See \cite[Construction 6.1]{Katz-Moore} for a proof that generalizes this result to work for all sequences with complex terms.

We want to investigate the correlation of Golay--Rudin--Shapiro sequences.
There are two main criteria for determining how low the correlation is: a worst-case ($\ell^\infty$) criterion and mean square ($\ell^2$) criterion.
For the worst case criterion, we define the {\it peak crosscorrelation of sequences $f$ and $g$} as $$\PCC(f,g)= \underset{s \in \Z}{\max }|C_{f,g}(s)|,$$
and the {\it peak sidelobe level of $f$} as $$\PSL(f)= \underset{s \in \Z \smallsetminus \{0\}}{\max }|C_{f,f}(s)|.$$
The mean square criterion is called the {\it demerit factor}, which gives the sum of the squared magnitudes of the correlation values over all shifts (excluding shift zero if autocorrelation is in question) when our sequences are normalized to have unit Euclidean norm.
The reciprocal of a demerit factor is called a {\it merit factor}.
The merit factor for autocorrelation was defined by Golay, who coined the term ``factor" in \cite[p.~450]{Golay-72}, and later the full name ``merit factor'' in \cite[p.~460]{Golay-75}. The term ``demerit factor'' for autocorrelation appears in the work of Sarwate \cite[p.~102]{Sarwate}, while for crosscorrelation, terminology both for ``merit factor'' and ``demerit factor'' appears in \cite[p.~5237]{Katz-16}, although mean square crosscorrelation had been studied much earlier (cf.~\cite[eqs.~(31),(38)]{Sarwate}).

Sarwate \cite[eqs.~(13),(38)]{Sarwate} showed that if one randomly selects binary sequences of length $\ell$, the mean autocorrelation demerit factor is $1-1/\ell$ and the mean crosscorrelation demerit factor for pairs is $1$.
It is known that the variance of the autocorrelation demerit factor approaches $0$ as $\ell$ tends to infinity \cite[pp.~1469--1470]{Borwein-Lockhart}.
Thus, a randomly selected long binary sequence is very likely to have an autocorrelation demerit factor close to $1$.
Calculations of the $L^4$ norm on the complex unit circle of the associated polynomials by Littlewood \cite[Problem 19]{Littlewood} show that the autocorrelation demerit factor of the Rudin--Shapiro sequences tends to $1/3$ as their length tends to infinity; this was rediscovered and explicitly realized to be a merit factor result by H\o holdt, Jensen, and Justesen in \cite[Theorem 2.3]{Hoholdt-Jensen-Justesen}.
The results of Katz and Moore \cite[Theorem 1.1]{Katz-Moore} then show that the demerit factor for the crosscorrelation between a Rudin--Shapiro sequence and its companion tends to $2/3$ as the length of the sequences tends to infinity.
Thus, Rudin--Shapiro sequences have significantly lower mean square autocorrelation and crosscorrelation than randomly selected binary sequences.

One should also consider how the Rudin--Shapiro sequences and their relatives rate in terms of worst-case measures of correlation (peak sidelobe level and peak crosscorrelation).
In 1985, H\o holdt, Jensen, and Justesen \cite[p.~552]{Hoholdt-Jensen-Justesen} obtained a bound on the peak sidelobe level of the $n$th Rudin--Shapiro sequence, which we now state.
\begin{theorem}[H\o holdt--Jensen--Justesen, 1985]\label{Cassie}
If $n$ is a nonnegative integer  and $x_n$ is the $n$th Rudin-Shapiro sequence, then there is some $A>0$ such that $$\PSL(x_n) \leq A(1.842626 \ldots)^n,$$
where $1.842626 \ldots$ is the unique positive real root of $X^4-3 X-6$.
\end{theorem}
In fact, H\o holdt--Jensen--Justesen's proof works for a larger family of sequences than the Rudin--Shapiro sequences, where the extra generality comes from allowing one to choose whether or not to swap $x_n$ and $y_n$ with each other at each step of the recursion.
Further work has been done to set lower and upper bounds on $\PSL(x_n)$ for the $n$th Rudin--Shapiro sequence, $x_n$, in works of Taghavi \cite{Taghavi-1996,Taghavi-1997,Taghavi-2007} and Allouche, Choi, Denise, Erd\'elyi, and Saffari \cite{Allouche-Choi-Denise-Erdelyi-Saffari}, culminating in the work of Choi \cite[Theorem 1.1]{Choi}, who gives the bound
\[
0.27771487\ldots (1+o(1)) \alpha_0^n \leq   \PSL(x_n) \leq (3.78207844\ldots) \alpha_0^n,
\]
where $\alpha_0=1.658967081916\ldots$ is the real root of  $X^3+X^2-2 X-4$, and the other numerical constants are approximations of numbers not explicitly defined in the theorem.
In \cite[Theorem 1.1]{Choi}, there is also a bound on the crosscorrelation $C_{y_n,x_n}(s)$ values for the $n$th Rudin--Shapiro companion sequence $y_n$ with the $n$th Rudin--Shapiro sequence $x_n$ at shifts $s$ with $s>0$,
\[
0.46071984\ldots (1+o(1)) \alpha_0^n \leq   C_{y_n,x_n}(s) \leq (3.78207844\ldots) \alpha_0^n,
\]
but there is no bound for shifts $s <0$, so we do not get a bound on $\PCC(y_n,x_n)$.\footnote{A pair of formulae appears both in \cite[Theorem 1.1]{Choi} and \cite[Theorem 1]{Allouche-Choi-Denise-Erdelyi-Saffari}: $b_0=2-L_n$ and $b_k=\conj{b_{-k}}$, which become $C_{y_n,x_n}(0)=2-2^n$ and $C_{y_n,x_n}(k)=\conj{C_{y_n,x_n}(-k)}$ if translated into the notation of this paper.  If true, these would supply the rest of the crosscorrelation values, but in fact $C_{y_n,x_n}(0)=0$ for all $n>1$ (see \cref{Vladwick} of this paper), and the second formula has many counterexamples, e.g., $x_1=1+z$ and $y_1=1-z$ so that $C_{y_1,x_1}(1)=-1\not=1=\conj{C_{y_1,x_1}(-1)}$, or $x_2=1+z+z^2-z^3$ and $y_2=1+z-z^2+z^3$ so that $C_{y_2,x_2}(1)=1\not=3=\conj{C_{y_2,x_2}(-1)}$.}

Our \cref{Gale} (proved as parts of \cref{Destiny} and \cref{George}) provides an improvement of \cref{Cassie} and all the more recent bounds on the peak sidelobe level of Rudin--Shapiro sequences.
\begin{theorem}\label{Gale}
If $n$ is a nonnegative integer and $x_n$ and $y_n$ the $n$th Rudin-Shapiro sequence and its companion, then $\PSL(x_n)=\PSL(y_n)$, and we have the upper bound
\[
\PSL (x_n) \leq 5 \alpha_0^{n-4}=(0.660113 \ldots)\alpha_0^n,
\]
where $\alpha_0=1.658967\ldots$ is the real root of  $X^3+X^2-2 X-4$.
This upper bound becomes an equality if $n=4$; otherwise it is a strict inequality.
We also have the lower bound
\[
\PSL(x_n) \geq \left[(0.382159...\ldots) - (0.421193\ldots)(0.935994\ldots)^n\right] \alpha_0^n,
\]
where $0.382159\ldots=(3\alpha_0^2+21\alpha_0+2)/118$, the quantity $0.421193\ldots$ is equal to $\sqrt{(6\alpha_0^2 + 6\alpha_0 -16)/59}$, and $0.935994\ldots=\sqrt{(\alpha_0^2-1)/2}$.
\end{theorem}
This shows that \(0.382159 \ldots \leq \limsup_{n \to \infty} \PSL(x_n)/\alpha_0^n \leq 0.660113\ldots\) (cf.~\cref{Sergio}), so that one cannot lower the exponential base our upper bound without rendering it invalid for almost every $n$.
Furthermore, if one maintains the base $\alpha_0$, then the fact that the upper bound is met for $n=4$ shows that one cannot lower its constant prefactor without rendering the bound invalid for at least one $n$.
The lower bound in \cref{Gale} is not very strong for small values of $n$.
Based on computation of $\PSL(x_n)$ for small values of $n$, we also obtain (as \cref{Vanessa}) the following lower bound, valid for all $n>0$ and stronger than the lower bound in \cref{Gale} for $n\leq 42$:
\[
\PSL(x_n)=\PSL(y_n) \geq 133991557 \alpha_0^{n-39}=(0.357605\ldots) \alpha_0^n,
\]
with exact equality when $n=39$, and strict inequality otherwise.
The prefactor of this lower bound is more than half the prefactor of the upper bound from \cref{Gale}.

These results were obtained using a relation between the autocorrelation for the Rudin--Shapiro sequence $x_n$ and the crosscorrelation for $x_{n-1}$ and its companion sequence $y_{n-1}$.
So in fact, we first proved (as parts of \cref{Destiny} and \cref{George}) the following analogue of \cref{Gale} for crosscorrelation.
\begin{theorem}\label{Gabby}
If  $n$ is a nonnegative integer and $x_n$ and $y_n$ the $n$th Rudin-Shapiro sequence and its companion, then we have the upper bound
\[
\PCC (x_n, y_n)\leq 5 \alpha_0^{n-3}=( 1.095107\ldots)\alpha_0^n,
\]
where $\alpha_0=1.658967\ldots$ is the real root of  $X^3+X^2-2 X-4$.
This upper bound becomes an equality if $n=3$; otherwise it is a strict inequality.
We also have the lower bound
\[
\PCC(x_n,y_n) \geq \left[(0.633990\ldots) - (0.654022\ldots)(0.935994\ldots)^n\right] \alpha_0^n,
\]
where $0.633990\ldots=(9 \alpha_0^2 + 4 \alpha_0 + 6)/59$, the quantity $0.654022\ldots$ is equal to $2\sqrt{(-4\alpha_0^2 + 2\alpha_0 + 14)/59}$, and $0.935994\ldots=\sqrt{(\alpha_0^2-1)/2}$.
\end{theorem}
As with the upper bound in \cref{Gale}, we cannot lower the exponential base of the upper bound without invalidating the bound for almost all $n$ (see \cref{Sergio}), and if we maintain the base, then we cannot lower the prefactor without the bound failing for at least one $n$.
We also have a computation-based lower bound on $\PCC(x_n,y_n)$ from \cref{Vanessa}, which is valid for all $n \in \N$ and stronger than the lower bound in \cref{Gabby} for $n\leq 41$:
\[
\PCC(x_n,y_n) \geq 133991557 \alpha_0^{n-38} = (0.593256\ldots) \alpha_0^n,
\]
with exact equality when $n=38$, and strict inequality otherwise.
The prefactor of this lower bound is more than half the prefactor of the upper bound from \cref{Gabby}.

In \cref{Generic Destiny} (restated here as \cref{Tommy}) we also prove analogous bounds for the more general class of Golay--Rudin--Shapiro sequences from \cref{Stan}.
To obtain a valid constant prefactor for the bounds, one needs to know the correlation values for the seed pair, $(x_0,y_0)$, which is the origin of the recursion.
\begin{theorem}\label{Tommy}
Suppose that $n$ is a nonnegative integer and $\ell_n$, $x_n$, and $y_n$ are as in \cref{Stan}, and that $K=9\alpha_0^{-4} \PCC(x_0,y_0) +18 \alpha_0^{-5} \PSL(x_0)$, where $\alpha_0=1.658967\ldots$ is the real root of  $X^3+X^2-2 X-4$.
Then $\PCC(x_n,y_n)\leq K \alpha_0 ^n$, and if $n>0$, then $\PSL(x_n)= \PSL(y_n)\leq K \alpha_0 ^{n-1}$.
\end{theorem}
As with the bounds in Theorems \ref{Gale} and \ref{Gabby}, we cannot lower the exponential base of either of these bounds without invalidating the bound for almost all $n$ (see \cref{Generic Sergio}).

This paper concerns aperiodic correlation, but here we briefly state the significance of these bounds for periodic correlation.
Suppose that $f=(f_0,f_1,\ldots,f_{k-1})$ and $g=(g_0,g_1,\ldots,g_{k-1})$, and we wish these to represent periodic sequences with period $k$.
Then for $s \in \{0,1,\ldots,k-1\}$, the {\it periodic crosscorrelation of $f$ with $g$}, written $\PC_{f,g}(s)$, equals $C_{f,g}(s)+C_{f,g}(s-k)$.
In particular, note that if we are concerned with periodic autocorrelation at shift $0$, then $\PC_{f,f}(0)=C_{f,f}(0)$, since $C_{f,f}(-k)=0$.
Since every periodic correlation is the sum of two aperiodic correlation values, we see that the periodic analogues of peak crosscorrelation and peak sidelobe level are clearly bounded above by twice the aperiodic bound:
\begin{align*}
\max_{0 \leq s < k} |\PC_{f,g}(s)| & \leq 2 \PCC(f,g) \\
\max_{0 < s < k} |\PC_{f,f}(s)| & \leq 2 \PSL(f).
\end{align*}
Thus, our upper bounds from Theorems \ref{Gale}--\ref{Tommy} all yield upper bounds on periodic correlation.  For example, if the Rudin--Shapiro sequences $x_n$ and $y_n$ are regarded as sequences of period $2^n$, then \cref{Gabby} shows that $\max_{0 \leq s < 2^n} |\PC_{x_n,y_n}(s)| \leq 10 \alpha_0^{n-3}$ for every $n \in \N$.
Since this bounding technique does not take account of possible cancellation between the two summands $C_{x_n,y_n}(s)$ and $C_{x_n,y_n}(s-2^n)$ in the expression for $\PC_{x_n,y_n}(s)$, our bounds on periodic correlation may not be tight, but analysis and refinement of periodic bounds is beyond the scope of this paper.

This paper is organized as follows.
In \cref{GRS}, the relevant background material will be explained, including a brief exposition in \cref{Ollie} on how to use solely rational arithmetic to verify equations and inequalities involving the algebraic number $\alpha_0$, which appears in our bounds above.
In \cref{GRS} we also begin to calculate autocorrelation and crosscorrelation of Golay--Rudin--Shapiro sequences in terms of sequences from earlier stages of the recursion.
In \cref{Iterating}, we iterate this recursion multiple times to obtain a bound on the crosscorrelation of our Golay pairs for the vast majority of shifts, but there is a small set of shifts for which we cannot obtain the desired bounds.
We provide the results needed to close this gap in \cref{Shifts}, where we prove a bound on the crosscorrelation of our Golay pairs at particular shifts that follow a special recursion.
We then combine the bounds from \cref{Iterating} and \cref{Shifts} to obtain a bound for all shifts in \cref{Final Bounds}.
We use multiple methods because there does not appear to be a single technique that can be practically applied to achieve our bounds.

\section{Preliminaries}\label{GRS}

In this paper, $\N$ denotes the set of nonnegative integers and $\Z^+$ denotes the set of strictly positive integers.
All the definitions and notations in the Introduction remain in force for the rest of the paper.
In particular, the reader should recall that we always identify a sequence $(f_0,f_1,\ldots,f_{k-1})$ of complex numbers with the polynomial $f(z)=\sum_{j \in \Z} f_j z^j$ using the convention that $f_j=0$ if $j\not\in\{0,1,\ldots,k-1\}$.
Furthermore, $\ell_n$, $x_n$, and $y_n$ are always as specified in \cref{Stan}.
That is, $\ell_0$ is a positive integer and $x_0$ and $y_0$ are nonzero polynomials in $\C[z]$ of degree less than $\ell_0$ such that $(x_0,y_0)$ is a Golay complementary pair with $C_{x_0,x_0}(0)=C_{y_0,y_0}(0)$.
Also $\ell_n=2^n \ell_0$ for each $n \in \N$, and we recursively define an infinite sequence of pairs $\{(x_n,y_n)\}_{n \in \N}$ of polynomials by the rules $x_n(z)=x_{n-1}(z)+z^{\ell_{n-1}} y_{n-1}(z)$ and $y_n(z)=x_{n-1}(z)-z^{\ell_{n-1}} y_{n-1}(z)$.
That is, $(x_n,y_n)$ is the $n$th Golay pair obtained from $\ell_0$ and the seed pair $(x_0,y_0)$ via the Golay--Rudin--Shapiro recursion.

\subsection{Fundamental facts}\label{Benjamin}
Throughout this paper, we let $ \C[z,z^{-1}] $ denote the ring of Laurent polynomials with coefficients from $\C$.
Because there is a strong connection between sequences and Laurent polynomials on the complex unit circle, if $ a(z)=\sum_{j\in \Z}a_jz^j $, then we use $ \conj{a(z)} $ as a shorthand for $\sum_{j\in \Z}\conj{a_j}z^{-j}$, and we also use $|a(z)|^2$ as a shorthand for $a(z) \conj{a(z)}$.
The crosscorrelation of two sequences (see \eqref{Tatiana}) can be expressed as a coefficient of a particular Laurent polynomial.
\begin{lemma}\label{David}
If $f$ and $g$ are sequences, then
\[
f(z) \conj{g(z)} = \sum_{s \in \Z} C_{f,g}(s) z^s.
\]
\end{lemma}
\begin{proof}
Write $f(z)=\sum_{j \in \Z} f_j z^j$ and $g(z)=\sum_{j \in \Z} g_j z^j$.
The coefficient of $z^s$ in $f(z)\conj{g(z)}$ is $\sum_{j \in \Z} f_{j+s} \conj{g_j}$, which is precisely $C_{f,g}(s)$.
\end{proof}
\begin{corollary}\label{Bob}
If $f$ and $g$ are sequences and $s \in \Z$, then $C_{g,f}(s)=\conj{C_{f,g}(-s)}$.
\end{corollary}

\begin{corollary}\label{Brady}
If $f$ and $g$ are sequences, then $\PCC(f,g) =\PCC(g,f)$.
\end{corollary}
We do not need to keep track of separate values of peak sidelobe level for the two sequences in a Golay pair, because they are always identical.
\begin{lemma}\label{Sylvester}
If $(f,g)$ is a Golay pair, then $\PSL(f)=\PSL(g)$.
\end{lemma}
\begin{proof}
Since $(f,g)$ is a Golay pair, we know that  $|C_{f,f}(s)|=|-C_{g,g}(s)|=|C_{g,g}(s)|$ for all nonzero $s \in \Z$. Therefore, $\PSL(f)= \PSL(g)$.
\end{proof}

\subsection{Correlation recursions}
We now prove some basic results on how the correlation values for Golay--Rudin--Shapiro sequences can be calculated from correlation values of sequences appearing in earlier stages of the recursion.
First we record an important result on the degrees of our Golay--Rudin--Shapiro sequences.
\begin{lemma}\label{Claudia}
We have $\deg(x_n), \deg(y_n) < \ell_n$ for each $n \in \N$.
If $\deg(y_0)=\ell_0-1$, then $\deg(x_n)=\deg(y_n)=\ell_n-1$ for each $n \in \N$.
If $x_0$ and $y_0$ are binary sequences of length $\ell_0$, then $x_n$ and $y_n$ are binary sequences of length $\ell_n$ for each $n \in \N$. 
\end{lemma}
\begin{proof}
These all follow easily by induction from the equations $\ell_n = 2^n\ell_0$, $x_n(z)=x_{n-1}(z)+z^{\ell_{n-1}} y_{n-1}(z)$, and $y_n(z)=x_{n-1}(z)-z^{\ell_{n-1}} y_{n-1}(z)$ of the Golay--Rudin--Shapiro recursion (\cref{Stan}).
\end{proof}
In view of \cref{David}, the following tells us how to relate the correlation values for the pair $(x_n,y_n)$ to those of the pair $(x_{n-1},y_{n-1})$.
\begin{lemma}\label{Matilda}
For any $n \geq 1$, we have
\begin{align*}
|x_{n}|^2 & =|x_{n-1}|^2 + |y_{n-1}|^2+ z^{- \ell_{n-1}}x_{n-1} \conj{y_{n-1}} +  z^{\ell_{n-1}} y_{n-1} \conj{x_{n-1}}, \\
|y_{n}|^2 & =|x_{n-1}|^2+|y_{n-1}|^2-z^{-\ell_{n-1}}x_{n-1}\conj{y_{n-1}}-z^{\ell_{n-1}}y_{n-1}\conj{x_{n-1}},\\
x_{n}\conj{y_{n}} & =|x_{n-1}|^2-|y_{n-1}|^2- z^{-\ell_{n-1}}x_{n-1}\conj{y_{n-1}}+z^{\ell_{n-1}}y_{n-1}\conj{x_{n-1}}, \text{ and}\\
y_{n}\conj{x_{n}} & =|x_{n-1}|^2-|y_{n-1}|^2+z^{-\ell_{n-1}}x_{n-1}\conj{y_{n-1}}-z^{\ell_{n-1}}y_{n-1}\conj{x_{n-1}}.
\end{align*}
\end{lemma}
\begin{proof}
All of these are direct consequences of $x_{n}=x_{n-1}+z^{\ell_{n-1}} y_{n-1}$ and $y_{n}=x_{n-1}-z^{\ell_{n-1}} y_{n-1}$ from \cref{Stan}.
\end{proof}

\begin{corollary}\label{Kelty}
For each $n \in \N$, the pair $(x_n,y_n)$ is a Golay complementary pair  with $|x_n|^2+|y_n|^2=2^n(|x_0|^2+|y_0|^2)$, which is a constant. 
\end{corollary}
\begin{proof}
This follows by induction once one adds the first two equations of \cref{Matilda}.
\end{proof}
\begin{lemma}\label{Lamar}  For every $n \geq 1$, we have
\begin{align*}
|x_{n}|^2 & = 2^{n-1}(|x_0|^2+|y_0|^2)+z^{-\ell_{n-1}}x_{n-1}\conj{y_{n-1}}+z^{\ell_{n-1}}y_{n-1}\conj{x_{n-1}}, \text{ and}\\
|y_{n}|^2 & = 2^{n-1}(|x_0|^2+|y_0|^2)-z^{-\ell_{n-1}}x_{n-1}\conj{y_{n-1}}-z^{\ell_{n-1}}y_{n-1}\conj{x_{n-1}}.
\end{align*}
For every $n \geq 2$, we have 
\begin{align}
\begin{split}\label{Reginald}
x_{n}\conj{y_{n}} = z^{\ell_{n-1}}y_{n-1}\conj{x_{n-1}} & -z^{-\ell_{n-1}}x_{n-1}\conj{y_{n-1}} \\ & +2 z^{\ell_{n-2}}y_{n-2}\conj{x_{n-2}} + 2 z^{-\ell_{n-2}}x_{n-2}\conj{y_{n-2}}
\end{split}
\end{align} 
and 
\begin{align*}
\begin{split}
y_{n}\conj{x_{n}} =z^{-\ell_{n-1}}x_{n-1}\conj{y_{n-1}} & - z^{\ell_{n-1}}y_{n-1}\conj{x_{n-1}} \\
& + 2 z^{\ell_{n-2}}y_{n-2}\conj{x_{n-2}} + 2 z^{-\ell_{n-2}}x_{n-2}\conj{y_{n-2}}.
\end{split}
\end{align*}
\end{lemma}
\begin{proof}
The first two assertions follow immediately from \cref{Matilda} and \cref{Kelty}.
The last two follow from the third and fourth equations of \cref{Matilda} by replacing instances of $|x_{n-1}|^2$ and $|y_{n-1}|^2$ on the right-hand sides with the expressions for $|x_{n-1}|^2$ and $|y_{n-1}|^2$ furnished by the first two equations of \cref{Matilda} (when one substitutes $n-1$ for $n$).
\end{proof}
\begin{definition}
We call the recursion in \eqref{Reginald} the {\it fundamental crosscorrelation recursion}.
\end{definition}
The following result shows that understanding the peak sidelobe levels of Golay--Rudin--Shapiro sequences is equivalent to understanding their peak crosscorrelations.
In fact, we will find it more convenient to track the peak crosscorrelation and then deduce the equivalent result for peak sidelobe level.
\begin{lemma}\label{Kenneth}
For $n \in \Z^+$, we have $\PSL(x_n)=\PCC(x_{n-1},y_{n-1})$.
\end{lemma}
\begin{proof}
Recall from \cref{Lamar} that we have 
\[
|x_{n}|^2 = \ell_{n-1}(|x_0|^2+|y_0|^2)+z^{-\ell_{n-1}}x_{n-1}\conj{y_{n-1}}+z^{\ell_{n-1}}y_{n-1}\conj{x_{n-1}}.
\]
To get the correlation of $x_n$ with itself at shift $s$, we can read off the coefficient of $z^s$ in this. Since $\deg x_{n-1}, \deg y_{n-1} < \ell_{n-1}$ by \cref{Claudia}, the penultimate term furnishes only strictly negative powers of $z$ (whose nonzero coefficients are the precisely the collection of nonzero crosscorrelation values for $x_{n-1}$ with $y_{n-1}$ by \cref{David}), while the final term furnishes only strictly positive powers of $z$ (whose nonzero coefficients are precisely the collection of nonzero crosscorrelation values for $y_{n-1}$ with $x_{n-1}$ by \cref{David}), and note that $\ell_{n-1} (|x_0|^2 + |y_0|^2)$ is just a constant.
Therefore,
\[
\PSL(x_n) = \max \{\PCC(x_{n-1},y_{n-1}), \PCC(y_{n-1},x_{n-1})\}=\PCC(x_{n-1},y_{n-1}),
\]
where the second equality follows by \cref{Brady}.
\end{proof}

\begin{lemma}\label{Geoff} If $n \geq 2$, then
\[
C_{x_{n},y_{n}}(s)= \begin{cases}
\conj{C_{x_{n-1},y_{n-1}}(\ell_{n-1}-s) }+2\conj{C_{x_{n-2},y_{n-2}}(\ell_{n-2}-s)}& \text{if $s>0$,}\\
-C_{x_{n-1},y_{n-1}}(\ell_{n-1}+s) +2C_{x_{n-2},y_{n-2}}(\ell_{n-2}+s)& \text{if $s<0$.}
\end{cases}
\]
\end{lemma}
\begin{proof}
Suppose that $s>0$. Read off the coefficient of $z^s$ in the fundamental crosscorrelation recursion \eqref{Reginald} given in \cref{Lamar}. The bounds on degree in \cref{Claudia} show that positive powers of $z$ occur only in the first and third terms of the right-hand side of \eqref{Reginald}, and then apply \cref{Bob} to obtain the first relation. On the other hand, if $s<0$, note that negative powers of $z$ occur only in the second and fourth terms on the right-hand side of \eqref{Reginald}.
\end{proof}
In \cite[Theorem 2.1]{Hoholdt-Jensen-Justesen} H\o holdt, Jensen, and Justesen prove that the autocorrelation of any Rudin-Shapiro sequence (or its companion sequence) at a nonzero even shift is equal to zero.
One has a similar result for the crosscorrelation between any Rudin-Shapiro sequence and its companion.
\begin{lemma}\label{Vladwick}
Suppose that $x_0=y_0=\ell_{0}=1$.
Then for every $n \in \N$, the Laurent polynomials $ |x_n|^2 $, $ |y_n|^2$, $x_n \conj{y_n}$, and $y_n \conj{x_n}$ have no terms of even degree, except that $|x_n|^2$ and $|y_n|^2$ have nonzero constant terms equal to $2^n$, and $x_0 \conj{y_0}$ and $y_0 \conj{x_0}$ have nonzero constant terms equal to $1$.
\end{lemma}
\begin{proof}
One can check these statements directly for $n=0$ and $n=1$, and then prove the rest by induction using the formulae from \cref{Matilda}.
\end{proof}
\begin{remark}
If we begin with $\ell_0=1$ and $x_0=y_0=1$, then $x_n$ and $y_n$ are the $n$th Rudin--Shapiro sequence and its companion.
Then $x_1=1+z$ and $y_1=1-z$, so one can easily calculate that $C_{x_0,y_0}(0)=1$, $C_{x_1,y_1}(-1)=-1$, and $C_{x_1,y_1}(1)=1$.
Recall that $C_{x_n,y_n}(s)=0$ if $|s| \geq 2^n$ (because the sequences are of length $2^n$) and $C_{x_n,y_n}(s)=0$ when $n>0$ and $s$ is even (by \cref{Vladwick}).
From these facts, one can use \cref{Geoff} to compute any other value of $C_{x_n,y_n}(s)$.
On \cref{Destiny Table}, we display values of $C_{x_n,y_n}(s)$ for selected shifts $s$ for values of $n \leq 10$.
These particular values will be useful to us when we prove our bounds on peak crosscorrelation of Rudin--Shapiro sequences in \cref{Destiny}.
The table is constructed so that any value of $C_{x_n,y_n}(s)$ on the table with $n \geq 2$ can be obtained via \cref{Geoff} from other correlation values that are on the table (or values known from the facts stated earlier in this paragraph).
{\footnotesize
\begin{table}
\begin{center}
\caption{Crosscorrelation values $C_{n,s}=C_{x_n,y_n}(s)$ at selected shifts $s$ for Rudin--Shapiro sequences $x_n$ and their companions $y_n$}\label{Destiny Table}
\begin{tabular}{|c||rr|rr|rr|rr|rr|}
\hline
$n$ & $s$ & $C_{n,s}$ & $s$ & $C_{n,s}$ & $s$ & $C_{n,s}$ & $s$ & $C_{n,s}$ & $s$ & $C_{n,s}$ \\
\hline
\multirow{-1}{*}{$0$} & $0$ & $1$ & & & & & & & & \\
\rowcolor{tableshade}
\multirow{-1}{*}{$1$} & $-1$ & $-1$ & $1$ & $1$ & & & & & & \\
\multirow{-1}{*}{$2$} & $-3$ & $1$ & $-1$ & $1$ & $1$ & $3$ & $3$ & $-1$ & & \\
\rowcolor{tableshade}
& $-5$ & $-1$ & $-3$ & $-5$ & $-1$ & $3$ & $1$ & $1$ & $3$ & $1$ \\
\rowcolor{tableshade}
\multirow{-2}{*}{$3$} & $5$ & $1$ & & & & & & & & \\
& $-11$ & $5$ & $-7$ & $1$ & $-5$ & $1$ & $-3$ & $5$ & $3$ & $7$  \\
\multirow{-2}{*}{$4$} & $5$ & $3$ & $7$ & $3$ & $11$ & $-5$ & & & & \\
\rowcolor{tableshade}
& $-21$ & $-1$ & $-13$ & $-9$ & $-11$ & $-13$ & $-9$ & $3$ & $-5$ & $7$  \\
\rowcolor{tableshade}
\multirow{-2}{*}{$5$} & $5$ & $-3$ & $9$ & $9$ & $11$ & $-7$ & $13$ & $5$ & $21$ & $1$  \\
& $-43$ & $13$ & $-41$ & $-3$ & $-27$ & $13$ & $-23$ & $-7$ & $-21$ & $9$  \\
& $-11$ & $5$ & $11$ & $7$ & $19$ & $15$ & $21$ & $-5$ & $27$ & $7$  \\
\multirow{-3}{*}{$6$} & $41$ & $3$ & $43$ & $-13$ & & & & & & \\
\rowcolor{tableshade}
& $-85$ & $-9$ & $-53$ & $-9$ & $-45$ & $-33$ & $-43$ & $-21$ & $-23$ & $15$  \\
\rowcolor{tableshade}
& $-21$ & $-1$ & $21$ & $-27$ & $23$ & $21$ & $37$ & $21$ & $43$ & $-31$  \\
\rowcolor{tableshade}
\multirow{-3}{*}{$7$} & $53$ & $5$ & $85$ & $9$ & & & & & & \\
& $-107$ & $53$ & $-105$ & $-27$ & $-91$ & $5$ & $-85$ & $49$ & $-43$ & $-19$  \\
& $43$ & $-1$ & $75$ & $15$ & $85$ & $-13$ & $105$ & $15$ & $107$ & $-1$  \\
\multirow{-3}{*}{$8$} & $171$ & $-21$ & & & & & & & & \\
\rowcolor{tableshade}
& $-181$ & $-33$ & $-171$ & $-29$ & $85$ & $-83$ & $149$ & $-3$ & $151$ & $45$  \\
\rowcolor{tableshade}
\multirow{-2}{*}{$9$} & $171$ & $-55$ & $213$ & $-19$ & & & & & & \\
\multirow{-1}{*}{$10$} & $-363$ & $109$ & $-361$ & $-99$ & $-341$ & $153$ & $299$ & $-57$ & & \\
\hline
\end{tabular}
\end{center}
\end{table}
}
\end{remark}

\subsection{Correlation of seed pairs}
We record and prove some technical lemmata involving correlation values for the seed pair $(x_0,y_0)$ in \cref{Stan}.
These lemmata are used in \cref{Final Bounds} to obtain bounds on peak crosscorrelation of all Golay pairs produced according to \cref{Stan}.
\begin{lemma}\label{Demetrius} We have 	$\PCC(x_1,y_1)\leq  2 \PSL(x_0) + \PCC(x_0,y_0)$.
\end{lemma}
\begin{proof} Let $s \in \Z$.
By reading the coefficient of $z^s$ from the expression for $x_n \conj{y_n}$ in \cref{Matilda} (using $n=1$), we see that
\[
C_{x_1,y_1}(s)=C_{x_0,x_0}(s) - C_{y_0,y_0}(s) - C_{x_0,y_0}(s+\ell_0) + C_{y_0,x_0}(s-\ell_0),
\]
but since $x_0$ and $ y_0$ are polynomials of degree less than $\ell_0$, at most one of the last two terms can be nonzero.
Thus, using the triangle inequality, \cref{Brady}, and \cref{Sylvester}, we have $|C_{x_1,y_1}(s)| \leq 2 \PSL(x_0)+\PCC(x_0,y_0)$.
\end{proof}

\begin{lemma}\label{Brenda}
There is some integer $s$ with $|s| < \ell_0$ such that $C_{x_0,y_0}(s) \neq 0$. 
\end{lemma}
\begin{proof}  Since $x_0$ and $y_0$ are not zero,
write $x_0=a_hz^h+a_{h+1} z^{h+1} + \cdots + a_i z^i$ and $y_0=b_jz^j+ b_{j+1} z^{j+1} + \cdots + b_k z^k$ such that $ a_i, b_j \neq 0$. Then set $s=i-j$ and note that $C_{x_0,y_0}(s)$ is the coefficient of $z^{i-j}$ in
\[
x_0 \conj{y_0} = (a_hz^h+a_{h+1} z^{h+1} + \cdots + a_i z^i)(\conj{ b_k}z^{-k}+ \conj{b_{k-1}} z^{-k+1} + \cdots + \conj{b_j} z^{-j}),
\]
which is $a_i \conj{b_j} \neq 0$ by our assumption. Note that $0 \leq i,j < \ell_{0}$, so that $|s| < \ell_0$.
\end{proof}

\begin{lemma}\label{Patty}
If  $\gamma \in \C$ with $|\gamma| \neq 1$, then there is some integer $s$ with $|s| < \ell_0$ such that $C_{x_0,y_0}(s) \neq \gamma \conj{C_{x_0,y_0}(-s)}$. 
\end{lemma}
\begin{proof}
By \cref{Brenda}, we can choose an integer $t$ with $|t|<\ell_0$ such that $C_{x_0,y_0}(t) $ is nonzero. If $C_{x_0,y_0}(t)= \gamma \conj{C_{x_0,y_0}(-t)}$ and $C_{x_0,y_0}(-t)= \gamma \conj{C_{x_0,y_0}(-(-t))}$, then $C_{x_0,y_0}(t)= |\gamma |^2 C_{x_0,y_0}(t)$. This implies that $|\gamma|=1$, which is a contradiction.
\end{proof}

\subsection{Algebraic number theory}\label{Ollie}
The bounds that we prove in this paper (Theorems \ref{Gale}--\ref{Tommy}) use the unique real root, $\alpha_0$, of the polynomial $m(X)=X^3+X^2-2 X-4 \in \Q[X]$.
The rational roots theorem shows that this polynomial is irreducible in $\Q[X]$, and its discriminant is $-236$, so it has the single real root $\alpha_0$, and two conjugate non-real roots, which we call $\alpha_1$ and $\alpha_2$; $m(X)$ is the minimal polynomial over $\Q$ of its own three roots.
It is important for us to be able to perform arithmetic in the splitting field $K=\Q(\alpha_0,\alpha_1,\alpha_2)$ of this polynomial and to prove inequalities between pairs of real elements in this splitting field.
By the irreducibility of $m(X)$, we know that $[\Q(\alpha_0):\Q]=3$ and since $\Q(\alpha_0) \subseteq \R$, we know that $m(X)/(X-\alpha_0)$ does not factor in this field, so $[K:\Q(\alpha_0)]=2$, and so $[K:\Q]=6$.
Then $\{\alpha_0^i: 0 \leq i < 3\}$ is a $\Q$-basis of $\Q(\alpha_0)$ and $\{\alpha_0^i \alpha_1^j: 0 \leq i < 3, 0 \leq j < 2\}$ is a $\Q$-basis of $K$.
Furthermore $K\cap \R$ is a proper subfield of $K$ and includes $\Q(\alpha_0)$, but $[K:\Q(\alpha_0)]=2$, so we know that $K\cap\R=\Q(\alpha_0)$.
Thus, we want to know how to reduce arithmetic operations in $K=\Q(\alpha_0,\alpha_1,\alpha_2)$ and inequalities in $K\cap\R=\Q(\alpha_0)$ to arithmetic operations and inequalities that involve only rational numbers: in this way we can use computers to assist our calculations without making any rounding errors that would compromise the certainty of our claims.

The reduction of arithmetic in an algebraic number field to rational arithmetic is well known, but we give a brief summary of an algorithm for converting any rational expression composed of elements of $K$ into a $\Q$-linear combination of the $\Q$-basis $\{\alpha_0^i \alpha_1^j: 0 \leq i < 3, 0 \leq j < 2\}$ of $K$.
We begin with a procedure for terms expressed as $a(\alpha_0,\alpha_1,\alpha_2)$ where $a(X,Y,Z)$ is a polynomial in $\Q[X,Y,Z]$.
Since $m(X)=X^3+X^2-2 X-4=(X-\alpha_0)(X-\alpha_1)(X-\alpha_2)$, matching the quadratic and linear coefficients yields $\alpha_0+\alpha_1+\alpha_2=-1$ and $\alpha_0(\alpha_1+\alpha_2)+\alpha_1\alpha_2=-2$, so that $\alpha_1+\alpha_2=-\alpha_0-1$ and $\alpha_1\alpha_2=\alpha_0^2+\alpha_0-2$.
Thus, $\alpha_2=-\alpha_0-\alpha_1-1$ and $\alpha_1^2=\alpha_1(\alpha_1+\alpha_2)-\alpha_1\alpha_2=\alpha_1(-\alpha_0-1)-(\alpha_0^2+\alpha_0-2)$, and of course $\alpha_0^3=-\alpha_0^2+2\alpha_0+4$ since $\alpha_0$ is a root of $m(X)$.
So given any polynomial $a(X,Y,Z) \in \Q[X,Y,Z]$, we can (i) replace every $Z$ with $-X-Y-1$ to eliminate $Z$, then (ii) replace any term of the form $X^i Y^j$ in which $j \geq 2$ with $X^i (-X Y-Y-X^2-X+2)Y^{j-2}$ and keep doing so until the degree in $Y$ is less than $2$, and then (iii) replace every term $X^i Y^j$ in which $i \geq 3$ with $(-X^2+2 X+4) X^{i-3} Y^j$ and keep doing so until the degree in $X$ is less than $3$; we obtain a polynomial $b(X,Y)=\sum_{i=0}^{2} \sum_{j=0}^1 b_{i,j} X^i Y^j$ with each $b_{i,j} \in \Q$ and with $b(\alpha_0,\alpha_1)=a(\alpha_0,\alpha_1,\alpha_2)$.  This $b(X,Y)$ is called the {\it standard reduction} of $a(X,Y,Z)$.  Since $\{\alpha_0^i \alpha_1^j: 0 \leq i < 3, 0 \leq j < 2\}$ is a $\Q$-basis of $K$ and $\{\alpha_0^i: 0 \leq i < 3\}$ is a $\Q$-basis of $\Q(\alpha_0)=K\cap\R$, we see that $b(\alpha_0,\alpha_1) \in \R$ if and only if $Y$ does not appear in $b(X,Y)$.

In general, if we combine elements of $K$ using addition, subtraction, multiplication, and division, we obtain an element $k=c(\alpha_0,\alpha_1,\alpha_2)/d(\alpha_0,\alpha_1,\alpha_2)$ for some $c(X,Y,Z)$, $d(X,Y,Z) \in \Q[X,Y,Z]$ with $d(\alpha_0,\alpha_1,\alpha_2)\not=0$.
Since $\alpha_1$ and $\alpha_2$ are conjugates, we know that $d(\alpha_0,\alpha_1,\alpha_2) d(\alpha_0,\alpha_2,\alpha_1)$ is a positive real number.
So the standard reduction of $d(X,Y,Z) d(X,Z,Y)$ is some nonzero polynomial $e(X) \in \Q[X]$, and $k=c(\alpha_0,\alpha_1,\alpha_2) d(\alpha_0,\alpha_2,\alpha_1)/e(\alpha_0)$.
Since $e(X)$ is a nonzero polynomial of degree less than $3$ and $m(X)$ is cubic and irreducible over $\Q$, we have $\gcd(e(X),m(X))=1$, so the Euclidean algorithm over $\Q[X]$ furnishes $f(X),g(X) \in \Q[X]$ such that $e(X) f(X)+m(X) g(X)=1$, so that $e(\alpha_0) f(\alpha_0)=1$.
Therefore, if we let $h(X,Y) \in \Q[X,Y]$ be the standard reduction of $c(X,Y,Z) d(X,Z,Y) f(X)$, then $k=h(\alpha_0,\alpha_1)$.
Thus, we have shown how to reduce any expression which might appear in our algebraic manipulations within $K=\Q(\alpha_0,\alpha_1,\alpha_2)$ to a $\Q$-linear combination of the $\Q$-basis $\{\alpha_0^i \alpha_1^j: 0 \leq i < 3, 0 \leq j < 2\}$ of $K$.

To deduce inequalities involving elements of $K\cap\R=\Q(\alpha_0)$, we compute the difference between two elements of $\Q(\alpha_0)$ to produce an expression of the form $v=p + q \alpha_0 + r \alpha_0 ^2$, where $p,q,r \in \Q$, and then determine whether this difference is greater than, equal to, or less than $0$.
We define a function that is the key to a practical algorithm for comparing $v$ with $0$.
\begin{definition}[signifier]\label{Abby}
Let $v = p + q \alpha_0 +r \alpha_0 ^2$ where $p,q,r \in \Q$.
We define the {\it signifier of $v$}, written $\sig(v)$, to be
\[
\sig(v) = p^3-p^2q-2pq^2+4q^3+5p^2r-10pqr-4q^2r+12pr^2-8qr^2+16r^3.
\]
\end{definition}
Notice that the signifier is a rational number computed using solely rational arithmetic.
In \cref{James} below it is proved that the quantity $v=p + q \alpha_0 +r \alpha_0 ^2$ is zero if and only if $\sig(v)$ is zero, and if $v$ and $\sig(v)$ are not zero, then they have the same sign.  First, we prove a preliminary result.
\begin{lemma}\label{Priscilla}
Suppose that $v=p+q \alpha_0+r \alpha_0^2$ for some $p,q,r, \in \Q$ with at least one of $q$ or $r$ nonzero.
Then the minimal polynomial of $v$ over $\Q$ is the cubic polynomial $n(X)=X^3+s X^2+t X+ u$ with
\begin{align*}
s & = -3p + q - 5r \\
t & = 3p^2 - 2pq - 2q^2 + 10pr - 10qr + 12r^2 \\
u & = -p^3+p^2q+2pq^2-4q^3-5p^2r+10pqr+4q^2r-12pr^2+8qr^2-16r^3,
\end{align*}
and the splitting field of $n(X)$ over $\Q$ is $K$, the splitting field of $m(X)=X^3+X^2-2 X-4$ over $\Q$, so that $n(X)$ has one real root and two conjugate non-real roots.
\end{lemma}
\begin{proof}
Since $[\Q(\alpha_0):\Q]=3$, the set $\{1,\alpha_0,\alpha_0^2\}$ is a $\Q$-basis of $\Q(\alpha_0)$, and since at least one of $q$ or $r$ is nonzero, this means $v\in\Q(\alpha_0)\smallsetminus \Q$.
Since $[\Q(\alpha_0):\Q]=3$, this forces $v$ to be degree $3$ over $\Q$, and since $\Q(v)\subseteq \Q(\alpha_0)$, this forces $\Q(v)=\Q(\alpha_0)$.
We let $Y$ be an indeterminate and set $w(Y)=p+q Y+r Y^2 \in \Q[Y]$, and then one can check that the polynomial $n(w(Y))\in \Q[Y]$ can be factored as
$Y^3+Y^2-2 Y-4$
times
$Y^3 r^3 + 3 Y^2 q r^2 - Y^2 r^3 + 3 Y q^2 r - 2 Y q r^2 - 2 Y r^3 + q^3 - q^2 r - 2 q r^2 + 4 r^3$.  So then $n(w(\alpha_0))=0$ because $\alpha_0$ is a root of $Y^3+Y^2-2 Y-4$.
But $w(\alpha_0)=v$,
so $v$ is of degree $3$ over $\Q$ and satisfies the monic cubic polynomial $n(X)$. Therefore, $n(X)$ is the minimal polynomial of $v$ over $\Q$.

Let $L$ be the splitting field of $n(X)$ over $\Q$.
Note that $v \in \Q(\alpha_0) \subseteq K$, so one root of $n(X)$ lies in $K$.
Since $K$ is the splitting field of $m(X)$ over $\Q$, this means $K$ is normal over $\Q$, and therefore all roots of $n(X)$ must lie in $K$.
So $L \subseteq K$.
Clearly $\Q(\alpha_0)=\Q(v) \subseteq L$.
But $\Q(\alpha_0)=\Q(v)$ is not normal over $\Q$, because it contains the real root $\alpha_0$ of $m(X)$ but does not contain the other two non-real roots of $m(X)$.
So $L$ must be strictly larger than $\Q(\alpha_0)=\Q(v)$.
Since $[K:\Q(\alpha_0)]=2$, this forces $L=K$.
Since $n(X) \in \Q[x]$ and $L\not\subseteq \R$, one root of $n(X)$ is real and the other two are non-real conjugates.
\end{proof}
\begin{proposition}\label{James} 
Let $v\in \Q (\alpha_0)$. Then $v=0$ if and only if $\sig(v)=0$, $v>0$ if and only if $\sig(v)>0$, and $v<0$ if and only if $\sig(v)<0$.
\end{proposition}
\begin{proof}
Write  $v = p + q \alpha_0 +r \alpha_0 ^2$ with $p,q,r \in \Q$.
Because $\{1,\alpha_0,\alpha_0^2\}$ is a $\Q$-basis of $\Q(\alpha_0)$, we know that $v \in \Q$ if and only if both $q$ and $r$ are zero.
In this case, $\sig(v)=p^3$, and the claim is clear. 
	
So henceforth, we assume $v \not\in \Q$.  
Let $n(X)=X^3+sX^2+tX+u$ be the minimal polynomial over $\Q$ of $v \in \Q(\alpha_0) \smallsetminus \Q$, which has only one real root by \cref{Priscilla}.
So the graph of the function $Y=n(X)$ crosses the $X$-axis precisely once.
Furthermore, since $n(X)$ is monic and of degree $3$, if its graph crosses the $X$-axis at a negative $X$-value (i.e., the real root $v$ is negative), then its graph crosses the $Y$-axis at a positive value (i.e., the constant coefficient $u$ is positive).
Similarly, if the $v$ is positive, then $u$ is negative.
Compare $u$ in \cref{Priscilla} with $\sig(v)$ in \cref{Abby} to see that $-u=\sig(v)$, so $\sig(v)$ has the same the sign as $v$.
\end{proof}
We can use the signifier to produce rational approximations of real elements in $K=\Q(\alpha_0,\alpha_1,\alpha_2)$ (i.e., elements of $\Q(\alpha_0)$) by comparing with rational numbers.
For example, we learn that $\alpha_0= 1.658967\ldots$. We also use the signifier along with the reduction procedures detailed earlier in this section to see that $\alpha_0$ is the root of $m(X)=X^3+X^2-2 X-4$ with largest magnitude.
\begin{lemma}\label{Gilda}
We have $|\alpha_1/\alpha_0|=|\alpha_2/\alpha_0|=\sqrt{(\alpha_0^2-1)/2}=0.935994\ldots$.
\end{lemma}
\begin{proof}
Since $\alpha_1$ and $\alpha_2$ are complex conjugates, we know that $|\alpha_1/\alpha_0|^2=|\alpha_2/\alpha_0|^2=\alpha_1\alpha_2/\alpha_0^2$, which can be shown to be equal to $(\alpha_0^2-1)/2$ by the reduction methods outlined above.
One then uses the signifier methods to see that $(\alpha_0^2-1)/2$ lies strictly between the rational numbers $0.935994^2$ and $0.935995^2$.
\end{proof}
One consequence of this is that no nonzero power of a root of $m(X)$ can be rational.
\begin{corollary}\label{Ulysses}
If $n$ is a nonzero integer, and $j \in \{0,1,2\}$, then $\alpha_j^n$ is irrational.
\end{corollary}
\begin{proof}
Suppose for a contradiction that there were some nonzero integer $n$, some $j \in \{0,1,2\}$, and some $q \in \Q$ such that $\alpha_j^n=q$.
Then since $\alpha_0,\alpha_1,\alpha_2$ are Galois conjugates over $\Q$, we can apply automorphisms of $K$ over $\Q$ to show that $\alpha_k^n=q$ for every $k \in \{0,1,2\}$.
In particular, this shows that $\alpha_0$ and $\alpha_1$ must have the same absolute value, which contradicts \cref{Gilda}.
\end{proof}

\section{Bounds from iterating the fundamental crosscorrelation recursion}\label{Iterating}

In this section, we iterate the fundamental crosscorrelation recursion \eqref{Reginald} from \cref{Lamar} multiple times to obtain formulae expressing correlation values for Golay--Rudin--Shapiro sequences in terms of correlation values of sequences appearing several steps earlier in the recursion (see \cref{Bjorn}).
We then use the iterated recursive relations to deduce bounds on peak crosscorrelation of Golay pairs (see \cref{Ingrid}).
These bounds are almost enough to prove via induction the main results of this paper, but there is a small interval of shifts where these bounds are not strong enough to achieve our desired bounds (see \cref{Velma}, which is proved in \cref{Hildegard}).
We develop different bounds later in \cref{Shifts} to fill this gap.
Recall that throughout this paper $\ell_n$, $x_n$, $y_n$ are always as in \cref{Stan}: $\ell_0$ is a positive integer and $x_0$ and $y_0$ are nonzero polynomials in $\C[z]$ of degree less than $\ell_0$ such that $(x_0,y_0)$ is a Golay complementary pair with $C_{x_0,x_0}(0)=C_{y_0,y_0}(0)$, and for every $n \in\N$, we have $\ell_n=2^n \ell_0$ and $(x_n,y_n)$ is the $n$th Golay pair obtained from $\ell_0$ and the seed pair $(x_0,y_0)$ via the Golay--Rudin--Shapiro recursion: $x_n(z)=x_{n-1}(z)+z^{\ell_{n-1}} y_{n-1}(z)$ and $y_n(z)=x_{n-1}(z)-z^{\ell_{n-1}} y_{n-1}(z)$.

\subsection{Iterating the fundamental crosscorrelation recursion}\label{Bjorn}

For $t \in \Z^+$, we recursively define Laurent polynomials $A_t(z)=\sum_{j \in \Z} A_{t,j} z^j$, $B_t(z)=\sum_{j \in \Z} B_{t,j} z^j$, $\Gamma_t(z)=\sum_{j \in \Z} \Gamma_{t,j} z^j$, and $\Delta_t(z)=\sum_{j \in \Z} \Delta_{t,j} z^j$ by
\begin{equation}\label{Gordon}
\begin{aligned}
& A_1(z) = -z^{-1}        && \qquad A_{t+1}(z) =-A_t(z^2)+ B_t(z^2) +z \Gamma_t(z^2), \\
& B_1(z) =  1            && \qquad B_{t+1}(z) =z A_t(z^2)-z B_t(z^2) +\Delta_t(z^2), \\
& \Gamma_1(z) = 2 z^{-1}  && \qquad \Gamma_{t+1}(z) =2 A_t(z^2)+ 2 B_t(z^2), \\
& \Delta_1(z)=2          && \qquad \Delta_{t+1}(z) =2 z A_t(z^2)+ 2 z B_t(z^2).
\end{aligned}
\end{equation}
We use these Laurent polynomials to express crosscorrelations for Golay--Rudin--Shapiro sequences in terms of crosscorrelations of  sequences appearing several steps earlier in the recursion.
Recall from \cref{Benjamin} that because there is a strong connection between sequences and Laurent polynomials on the complex unit circle, if $ a(z)=\sum_{j\in \Z}a_jz^j $, then we always use $\conj{a(z)}$ as a shorthand for $\sum_{j\in \Z}\conj{a_j}z^{-j}$.
\begin{proposition}\label{Brent}
If $t$ is a positive integer and $n > t$, then
\begin{multline*}
x_n\conj{y_n} = A_t(z^{\ell_{n-t+1}}) z^{2\ell_{n-t-1}} x_{n-t} \conj{y_{n-t}}
+ B_t(z^{\ell _{n-t+1}}) z^{2\ell_{n-t-1}} y_{n-t} \conj{x_{n-t}} \\
+ \Gamma _t(z^{\ell _{n-t+1}}) z^{3\ell_{n-t-1}} x_{n-t-1} \conj{y_{n-t-1}}
+ \Delta_t(z^{\ell _{n-t+1}}) z^{\ell_{n-t-1}} y_{n-t-1}\conj{x_{n-t-1}}.
\end{multline*}
\end{proposition}
\begin{proof}
Proceed by induction on $t$.
For $t=1$ the values of $A_1$, $B_1$, $\Gamma_1$, and $\Delta_1$ from \eqref{Gordon} make the equation we need to prove identical to \eqref{Reginald} from \cref{Lamar}.

Now assume that the desired identity holds for some $t \geq 1$ and let $w=z^{\ell_{n-t-2}}$, so we have
\begin{align}
\begin{split}\label{Penelope}
x_n\conj{y_n} = A_t(w^8) w^4 & x_{n-t} \conj{y_{n-t}}
+ B_t(w^8) w^4 y_{n-t} \conj{x_{n-t}} \\
& + \Gamma _t(w^8) w^6 x_{n-t-1} \conj{y_{n-t-1}}
+ \Delta_t(w^8) w^2 y_{n-t-1}\conj{ x_{n-t-1} },
\end{split}
\end{align}
and \eqref{Reginald} with $n-t$ in place of $n$ becomes
\begin{align}
\begin{split}\label{Oswald}
x_{n-t}\conj{y_{n-t}}  = w^2 y_{n-t-1} \conj{x_{n-t-1}} & -w^{-2} x_{n-t-1}\conj{y_{n-t-1}} \\
& +2 (w y_{n-t-2}\conj{x_{n-t-2}} +  w^{-1} x_{n-t-2}\conj{y_{n-t-2}}).
\end{split}
\end{align}
Then we use \eqref{Oswald} to replace $x_{n-t} \conj{y_{n-t}}$ and its conjugate on the right-hand side of \eqref{Penelope} and rearrange to obtain
\begin{align*}
x_n\conj{y_n}
= & \left(-A_t(w^8)+ B_t(w^8) +w^4\Gamma_t(w^8)\right) w^2 x_{n-t-1} \conj{y_{n-t-1}}\\
& \quad + \left(w^4 A_t(w^8)-w^4 B_t(w^8) +\Delta_t(w^8)\right) w^2 y_{n-t-1} \conj{x_{n-t-1}} \\
& \quad + \left(2 A_t(w^8)+ 2 B_t(w^8)\right) w^3 x_{n-t-2} \conj{y_{n-t-2}} \\
& \quad + \left(2 w^4 A_t(w^8)+ 2 w^4 B_t(w^8)\right) w y_{n-t-2} \conj{x_{n-t-2}},
\end{align*}
and then use the recursion in \eqref{Gordon} to obtain the desired identity for $t+1$.
\end{proof}
\begin{corollary}\label{Justine}
Let $s \in \Z$ and $n, t\in \N$ with $0 < t < n$ and write $s=q \ell_{n-t+1}  +r$ with $q, r \in \Z$ and $0 \leq r < \ell_{n-t+1}$.
For any $m \in \N$ and $u \in \Z$, we define $C_{m,u}$ to be the coefficient of $z^u$ in $x_m \conj{y_m}$.
Then
\begin{align*}
C_{n,s}=A_{t,q} C_{n-t,r-\ell_{n-t}} & + B_{t,q} \conj{C_{n-t,-r+\ell_{n-t}}}\\
& + \Gamma_{t,q} C_{n-t-1,r-3 \ell_{n-t-1}}  + \Delta_{t,q}\conj{C_{n-t-1,-r+\ell_{n-t-1}}},
\end{align*}
\end{corollary}
\begin{proof}
We read coefficients from the formula in \cref{Brent}.
First observe that $x_m \conj{y_m}$ (and its conjugate) can have nonzero coefficients for $z^u$ only if $|u| < \ell_m$.
Thus, the four terms $z^{2\ell_{n-t-1}} x_{n-t} \conj{y_{n-t}}$, $z^{2\ell_{n-t-1}} y_{n-t} \conj{x_{n-t}}$, $z^{3 \ell_{n-t-1}} x_{n-t-1} \conj{y_{n-t-1}}$, and $z^{\ell_{n-t-1}} y_{n-t-1} \conj{x_{n-t-1}}$ can have nonzero coefficients for $z^u$ only if $0 < u < \ell_{n-t+1}$.
On the other hand, $A_t(z^{\ell_{n-t+1}})$, $B_t(z^{\ell_{n-t+1}})$, $\Gamma_t(z^{\ell_{n-t+1}})$, and $\Delta_t(z^{\ell_{n-t+1}})$ have nonzero coefficients only for terms $z^u$ where $u$ is a multiple of $\ell_{n-t+1}$.
So the only way to get a term with $z^s=z^{q \ell_{n-t+1}+r}$ on the right-hand side of the formula from \cref{Brent} is to multiply the $q$th coefficient of $A_t$ (resp., $B_t$, $\Gamma_t$, $\Delta_t$) with the coefficient of $z^r$ in $z^{2\ell_{n-t-1}} x_{n-t} \conj{y_{n-t}}$ (resp., $z^{2\ell_{n-t-1}} y_{n-t} \conj{x_{n-t}}$, $z^{3 \ell_{n-t-1}} x_{n-t-1} \conj{y_{n-t-1}}$, $z^{\ell_{n-t-1}} y_{n-t-1} \conj{x_{n-t-1}}$), that is, with the coefficient of $z^{r-2\ell_{n-t-1}}$ (resp., the conjugate of the coefficient of $z^{-r+2\ell_{n-t-1}}$, the coefficient of $z^{r-3\ell_{n-t-1}}$, the conjugate of the coefficient of $z^{-r+\ell_{n-t-1}}$) in $x_{n-t} \conj{y_{n-t}}$ (resp., $x_{n-t} \conj{y_{n-t}}$, $x_{n-t-1} \conj{y_{n-t-1}}$, $x_{n-t-1} \conj{y_{n-t-1}}$), which is just $C_{n-t,r-\ell_{n-t}}$ (resp., $\conj{C_{n-t,-r+\ell_{n-t}}}$, $C_{n-t-1,r-3\ell_{n-t-1}}$, $\conj{C_{n-t-1,-r+\ell_{n-t-1}}}$).
\end{proof}

\begin{remark}
We need to compute some values of $A_{t,j}$, $B_{t,j}$, $\Gamma_{t,j}$, and $\Delta_{t,j}$ that will be critical in later proofs, and we record these values in Table \ref{Charles}.
For $t \in \Z_+$ and $j \in \Z$, the initial conditions from \eqref{Gordon} tell us that
\begin{equation}\label{Helen}
\begin{tabular}{ll}
$A_{1,j}=
\begin{cases}
-1 & \text{if $j=-1$,} \\
0 & \text{otherwise,}
\end{cases}$
&
$B_{1,j}=
\begin{cases}
1 & \text{if $j=0$,} \\
0 & \text{otherwise,}
\end{cases}$
\\
\\
$\Gamma_{1,j}=
\begin{cases}
2 & \text{if $j=-1$,} \\
0 & \text{otherwise,}  
\end{cases}$
&
$\Delta_{1,j}=
\begin{cases}
2 & \text{if $j=0$,} \\
0 & \text{otherwise.}
\end{cases}$
\end{tabular}
\end{equation}
The recursion in \eqref{Gordon} tells us that for every $t > 0$, we have
\begin{align}
\begin{split}\label{Ursula}
A_{t+1,j}= & \begin{cases}
-A_{t,j/2}+B_{t,j/2} & \text{if $j$ is even,} \\
\Gamma_{t, (j-1)/2} & \text{if $j$ is odd,}
\end{cases}\\
B_{t+1,j}= & \begin{cases}
\Delta_{t,j/2} & \text{if $j$ is even,} \\
A_{t,(j-1)/2}-B_{t,(j-1)/2} & \text{if $j$ is odd,}
\end{cases}\\
\Gamma_{t+1,j}= & \begin{cases}
2A_{t,j/2}+2B_{t,j/2}& \text{if $j$ is even,} \\
0 & \text{if $j$ is odd,}
\end{cases}\\
\Delta_{t+1,j}= & \begin{cases}
0& \text{if $j$ is even,} \\
2A_{t,(j-1)/2}+2B_{t,(j-1)/2} & \text{ if $j$ is odd.}
\end{cases}
\end{split}
\end{align}
The values of $A_{t,j}$, $B_{t,j}$, $\Gamma_{t,j}$, or $\Delta_{t,j}$ on \cref{Charles} with index $t=1$ come from \eqref{Helen}, and those with index $t>1$ can be computed via \eqref{Ursula} from entries with index $t-1$ that appear earlier on the table.
{\footnotesize
\begin{table}
\begin{center}  
\caption{Selected Values of $A_{t,j}$, $B_{t,j}$, $\Gamma_{t,j}$, and $\Delta_{t,j}$}\label{Charles}
\begin{tabular}{|c||rrrrr|rrrrr|}
\hline
$t$ & $j$ & $A_{t,j}$ & $B_{t,j}$ & $\Gamma_{t,j}$ & $\Delta_{t,j}$ & $j$ & $A_{t,j}$ & $B_{t,j}$ & $\Gamma_{t,j}$ & $\Delta_{t,j}$ \\
\hline
\multirow{-1}{*}{1} & -1 & -1 & 0 & 2 & 0 & 0 & 0 & 1 & 0 & 2 \\ 
\rowcolor{tableshade} 
\multirow{-1}{*}{2} & -1 & 2 & -1 & 0 & -2 & 0 & 1 & 2 & 2 & 0 \\ 
& -2 & -3 & -2 & 2 & 0 & -1 & 0 & 3 & 0 & 2 \\ 
\multirow{-2}{*}{3} & 0 & 1 & 0 & 6 & 0 & 1 & 2 & -1 & 0 & 6 \\ 
\rowcolor{tableshade} 
& -4 & 1 & 0 & -10 & 0 & -3 & 2 & -1 & 0 & -10 \\ 
\rowcolor{tableshade} 
\multirow{-2}{*}{4} & 1 & 6 & 1 & 0 & 2 & 2 & -3 & 6 & 2 & 0 \\ 
& -6 & -3 & -10 & 2 & 0 & 2 & -5 & 2 & 14 & 0 \\ 
\multirow{-2}{*}{5} & 4 & 9 & 0 & 6 & 0 & 5 & 2 & -9 & 0 & 6 \\ 
\rowcolor{tableshade} 
& -12 & -7 & 0 & -26 & 0 & -11 & 2 & 7 & 0 & -26 \\ 
\rowcolor{tableshade} 
& 5 & 14 & -7 & 0 & -6 & 9 & 6 & 9 & 0 & 18 \\ 
\rowcolor{tableshade} 
\multirow{-3}{*}{6} & 10 & -11 & 6 & -14 & 0  & & & & & \\ 
& -23 & -26 & -7 & 0 & -14 & -22 & 5 & -26 & 18 & 0 \\ 
& 10 & -21 & -6 & 14 & 0 & 11 & 0 & 21 & 0 & 14 \\ 
\multirow{-3}{*}{7} & 18 & 3 & 18 & 30 & 0 & 21 & -14 & -17 & 0 & -10 \\ 
\rowcolor{tableshade} 
& -46 & 19 & -14 & -66 & 0 & -43 & 18 & 31 & 0 & -42 \\ 
\rowcolor{tableshade} 
& 21 & 14 & -15 & 0 & -54 & 37 & 30 & -15 & 0 & 42 \\ 
\rowcolor{tableshade} 
\multirow{-3}{*}{8} & 42 & -3 & -10 & -62 & 0  & & & & & \\ 
& -91 & -66 & 33 & 0 & 10 & -86 & 13 & -42 & 98 & 0 \\ 
\multirow{-2}{*}{9} & 42 & -29 & -54 & -2 & 0 & 74 & -45 & 42 & 30 & 0 \\ 
\rowcolor{tableshade} 
& -182 & 99 & 10 & -66 & 0 & -181 & 0 & -99 & 0 & -66 \\ 
\rowcolor{tableshade} 
\multirow{-2}{*}{10} & -171 & 98 & 55 & 0 & -58 & 149 & 30 & -87 & 0 & -6 \\
\hline
\end{tabular}
\end{center}
\end{table}
}
\end{remark}

We can use \cref{Justine} to rapidly compute crosscorrelation values.
For $n,t \in \N$ with $0 < t < n$, we compute the crosscorrelation spectrum of the pair $(x_n,y_n)$ from (i) the crosscorrelation spectrum of the pair $(x_{n-t},y_{n-t})$, (ii) the crosscorrelation spectrum of the pair $(x_{n-t-1},y_{n-t-1})$, and (iii) the values of $A_{t,q}, B_{t,q}, \Gamma_{t,q}, \Delta_{t,q}$ for $-2^{t-1} \leq q < 2^{t-1}$ (a simple inductive argument from \eqref{Helen} and \eqref{Ursula} shows that $A_{t,q}=B_{t,q}=\Gamma_{t,q}=\Delta_{t,q}=0$ for all other values of $q$).
The number of operations it takes to compute the full crosscorrelation spectrum of $(x_n,y_n)$ in this manner is a small multiple of $\ell_n=2^n\ell_0$, and the number of memory locations needed to store (i) and (ii) is a small multiple of $\ell_{n-t}=2^{n-t}\ell_0$, while the memory to store (iii) is a small multiple of of $2^t$.
Thus, if $x_n$ and $y_n$ are the Rudin--Shapiro sequences, and if we keep $t$ around $n/2$, then the number of memory locations needed to compute the crosscorrelation spectrum for $(x_n,y_n)$ in this manner is a small multiple of $2^{n/2}$, and the number of arithmetical operations needed is only a small multiple of the length $\ell_n=2^n$.
Since each crosscorrelation value can be computed in turn, this technique can be used to compute the peak crosscorrelation values for Golay--Rudin--Shapiro sequences that are too large to fit in the computer's memory.

We used such techniques to compute the peak crosscorrelation for the Rudin--Shapiro sequence pairs $(x_n,y_n)$ of length $2^n$ for $0 \leq n \leq 50$.
The results are in \cref{Aeneas}: for each $n$ it lists every shift $s$ that yields maximum magnitude crosscorrelation, that is, every $s$ such that $|C_{x_n,y_n}(s)|=\PCC(x_n,y_n)$, and for each such $s$, it shows the crosscorrelation value $C_{x_n,y_n}(s)$, so that one also knows the sign.
Once we know about the crosscorrelation of the Rudin--Shapiro pair $(x_n,y_n)$, then we use the first equation in \cref{Lamar} (replacing $n$ with $n+1$) to see that for $s>0$ we have $C_{x_{n+1},x_{n+1}}(s)=C_{y_n,x_n}(s-\ell_n)=C_{x_n,y_n}(2^n-s)$; since $C_{x_{n+1},x_{n+1}}(-s)=\conj{C_{x_{n+1},x_{n+1}}(s)}$ by \cref{Bob}, it is redundant to determine the autocorrelation at negative shifts.
Proceeding in this way, we determine which positive shifts $s$ yield maximum autocorrelation $C_{x_{n+1},x_{n+1}}(s)$, i.e., which $s$ with $s>0$ have $|C_{x_{n+1},x_{n+1}}(s)|=\PSL(x_{n+1},x_{n+1})$.
For convenience, these shifts and their corresponding autocorrelation values are shown on \cref{Dido}.
It is redundant to indicate the autocorrelation values for the companion sequence, $y_{n+1}$, since Rudin--Shapiro pairs are Golay pairs.
{\footnotesize \begin{table}\caption{Crosscorrelation values $C_{n,s}=C_{x_n,y_n}(s)$ for Rudin--Shapiro pairs, $(x_n,y_n)$, at all shifts $s$ such that $|C_{n,s}|=\PCC(x_n,y_n)$}\label{Aeneas}
\begin{center}
\begin{tabular}{|c|rr|rr||c|rr|}
\hline
$n$ & $s$ & $C_{n,s}$ & $s$ & $C_{n,s}$ & $n$ & $s$ & $C_{n,s}$ \\
\hline
0 & 0 & 1 & &                         & 26 & -22369613 & 342769 \\                       
\rowcolor{tableshade}
1 & -1 & -1 & 1 & 1                   & 27 & -44739243 & -640933 \\                     
2 & 1 & 3 & &                         & 28 & -89478451 & 860709 \\                      
\rowcolor{tableshade}
3 & -3 & -5 & &                       & 29 & -178956971 & -1624877 \\                   
4 & 3 & 7 & &                         & 30 & -357913941 & 2490985 \\                    
\rowcolor{tableshade}
5 & -11 & -13 & &                     & 31 & -715827885 & -4188609 \\                   
6 & 13 & 19 & &                       & 32 & -1431655765 & 7618449 \\                   
\rowcolor{tableshade}
7 & -45 & -33 & &                     & 33 & -2863311539 & -10688117 \\                 
8 & -107 & 53 & &                     & 34 & -5726623061 & 20617465 \\                  
\rowcolor{tableshade}
9 & -179 & -85 & &                    & 35 & -11453246123 & -29999429 \\                
10 & -341 & 153 & &                   & 36 & -22906492245 & 51521697 \\                 
\rowcolor{tableshade}
11 & -717 & -217 & &                  & 37 & -45812984491 & -90947021 \\                
12 & -1451 & 373 & &                  & 38 & -91625968979 & 133991557 \\                
\rowcolor{tableshade}
13 & -2867 & -557 & &                 & 39 & -183251937963 & -255886741 \\              
14 & -5453 & 961 & &                  & 40 & -366503875925 & 372089521 \\               
\rowcolor{tableshade}
15 & -10923 & -1717 & &               & 41 & -733007751851 & -668060317 \\              
16 & -22955 & 2445 & &                & 42 & -1466015503701 & 1099665689 \\             
\rowcolor{tableshade}
17 & -43691 & -4285 & &               & 43 & -2932031007403 & -1724813029 \\            
18 & -91733 & 6257 & &                & 44 & -5864062014805 & 3146759617 \\             
\rowcolor{tableshade}
19 & -174765 & -11153 & &             & 45 & -11728124029611 & -4701529197 \\           
20 & -349525 & 19041 & &              & 46 & -23456248059221 & 8491242153 \\            
\rowcolor{tableshade}
21 & -699059 & -28293 & &             & 47 & -46912496118443 & -13498854709 \\          
22 & -1398101 & 53321 & &             & 48 & -93824992236885 & 22289746385 \\           
\rowcolor{tableshade}
23 & -2796237 & -72905 & &            & 49 & -187649984473771 & -38672931645 \\         
24 & -5592403 & 129485 & &            & 50 & -375299968947541 & 59901979961 \\          
\rowcolor{tableshade}
25 & -11184811 & -214365 & &          &    & & \\
\hline
\end{tabular}
\end{center}
\end{table}}
{\footnotesize
\begin{table}\caption{Autocorrelation values $D_{n,s}=C_{x_n,x_n}(s)$ for Rudin--Shapiro sequences, $x_n$, at all shifts $s>0$ such that $|D_{n,s}|=\PSL(x_n)$}\label{Dido}
\begin{center}
\begin{tabular}{|c|rr|rr||c|rr|}
\hline
$n$ & $s$ & $D_{n,s}$ & $s$ & $D_{n,s}$ & $n$ & $s$ & $D_{n,s}$ \\
\hline
0 & all $s>0$ & 0 & &               & 26 & 44739243 & -214365 \\                   
\rowcolor{tableshade}                                                                    
1 & 1 & 1 & &                         & 27 & 89478477 & 342769 \\                    
2 & 1 & 1 & 3 & -1                    & 28 & 178956971 & -640933 \\                  
\rowcolor{tableshade}                                                                    
3 & 3 & 3 & &                         & 29 & 357913907 & 860709 \\                   
4 & 11 & -5 & &                       & 30 & 715827883 & -1624877 \\                 
\rowcolor{tableshade}                                                                    
5 & 13 & 7 & &                        & 31 & 1431655765 & 2490985 \\                 
6 & 43 & -13 & &                      & 32 & 2863311533 & -4188609 \\                
\rowcolor{tableshade}                                                                    
7 & 51 & 19 & &                       & 33 & 5726623061 & 7618449 \\                 
8 & 173 & -33 & &                     & 34 & 11453246131 & -10688117 \\              
\rowcolor{tableshade}                                                                    
9 & 363 & 53 & &                      & 35 & 22906492245 & 20617465 \\               
10 & 691 & -85 & &                    & 36 & 45812984491 & -29999429 \\              
\rowcolor{tableshade}                                                                    
11 & 1365 & 153 & &                   & 37 & 91625968981 & 51521697 \\               
12 & 2765 & -217 & &                  & 38 & 183251937963 & -90947021 \\             
\rowcolor{tableshade}                                                                    
13 & 5547 & 373 & &                   & 39 & 366503875923 & 133991557 \\             
14 & 11059 & -557 & &                 & 40 & 733007751851 & -255886741 \\            
\rowcolor{tableshade}                                                                    
15 & 21837 & 961 & &                  & 41 & 1466015503701 & 372089521 \\            
16 & 43691 & -1717 & &                & 42 & 2932031007403 & -668060317 \\           
\rowcolor{tableshade}                                                                    
17 & 88491 & 2445 & &                 & 43 & 5864062014805 & 1099665689 \\           
18 & 174763 & -4285 & &               & 44 & 11728124029611 & -1724813029 \\         
\rowcolor{tableshade}                                                                    
19 & 353877 & 6257 & &                & 45 & 23456248059221 & 3146759617 \\          
20 & 699053 & -11153 & &              & 46 & 46912496118443 & -4701529197 \\         
\rowcolor{tableshade}                                                                    
21 & 1398101 & 19041 & &              & 47 & 93824992236885 & 8491242153 \\          
22 & 2796211 & -28293 & &             & 48 & 187649984473771 & -13498854709 \\       
\rowcolor{tableshade}                                                                    
23 & 5592405 & 53321 & &              & 49 & 375299968947541 & 22289746385 \\        
24 & 11184845 & -72905 & &            & 50 & 750599937895083 & -38672931645 \\       
\rowcolor{tableshade}                                                                    
25 & 22369619 & 129485 & &            & 51 & 1501199875790165 & 59901979961 \\
\hline
\end{tabular}
\end{center}
\end{table}}

\subsection{Bounds on the peak crosscorrelation}\label{Ingrid}
Now we show how to use \cref{Justine} to obtain bounds on the peak crosscorrelation of  Golay--Rudin--Shapiro sequences from bounds for sequences appearing earlier in the recursion.
This method almost provides an inductive proof of our bounds in Theorems \ref{Gale}--\ref{Tommy}, but we find that at each step there are sets of exceptional shifts for which we cannot continue the inductive chain.
These exceptional sets (named $E_t$ in \cref{Velma} below) tend to become smaller and smaller as we iterate the fundamental crosscorrelation recursion more times (corresponding to larger $t$ in \cref{Justine}), but we were not able to eliminate entirely these exceptional shifts after iterating many times.
Therefore we deemed it impractical to attempt an inductive proof Theorems \ref{Gale}--\ref{Tommy} based solely on this method of repeated iteration, and instead turned to another technique (described in \cref{Shifts}) to fill in the gaps.
\begin{notation}\label{Nancy}
For the remainder of this section and for \cref{Hildegard}, the following notations, conventions and definitions are used. 
\begin{itemize}
\item For every $k \in \N$, we write $C_k$ as a shorthand for $x_k \conj{y_k}$, and we write the expansion of $C_k$ as $C_k(z)=\sum_{j \in \Z} C_{k,j} z^j$. 
\item We let $M_k=\max_{j\in \Z} |C_{k,j}| $, which is just the maximum magnitude crosscorrelation between $x_k$ and $y_k$, i.e., $\PCC(x_k,y_k)$.
\item We use interval notation for subsets of $\Z$ rather than $\R$, so that $(a,b)$ (resp., $(a,b]$, $[a,b)$, $[a,b]$) is the set of all $s \in \Z$ such that $a<s<b$ (resp., $a<s\leq b$, $a \leq s<b$, $a\leq s\leq b$).
\item If $S \subseteq \Z$ and $c$ is a positive integer, we write $S \cdot c$ to mean $\{sc: s \in S\}$. For example, $\left((1,2]\cup(3,4)\right)\cdot c=(c,2 c]\cup(3 c,4 c)$.
\end{itemize}
\end{notation}
Now we use \cref{Justine} to produce general bounds on peak crosscorrelation.
\begin{lemma}\label{Nellie}
Let $n, t \in \N$ with $0 < t < n$, let $s\in\Z$, and write $s=q \ell_{n-t+1} +r$ with $q,r \in \Z$ and $0 \leq r < \ell_{n-t+1}$.
Then
\[
|C_{n,s}| \leq
\begin{cases}
(|A_{t,q}|+|B_{t,q}|) M_{n-t} + |\Delta_{t,q}| M_{n-t-1} & \text{if $0 < r < \ell_{n-t}$}, \\
(|A_{t,q}|+|B_{t,q}|) M_{n-t} + |\Gamma_{t,q}| M_{n-t-1} & \text{if $\ell_{n-t} < r < \ell_{n-t+1}$},\\
(|A_{t,q}|+|B_{t,q}|) M_{n-t} & \text{if $r=\ell_{n-t}$},
\end{cases}
\]
and $C_{n,s}=0$ if $r=0$.	
\end{lemma}
\begin{proof}
By \cref{Justine}, the coefficient of $z^s$ in $x_n \conj{y_n}$ is
\begin{align}
\begin{split}\label{Mabel}
C_{n,s}=A_{t,q} C_{n-t,r-\ell_{n-t}} & + B_{t,q} \conj{C_{n-t,\ell_{n-t}-r}}\\
& +\Gamma_{t,q} C_{n-t-1,r-3 \cdot\ell_{n-t-1}}  + \Delta_{t,q}\conj{C_{n-t-1,\ell_{n-t-1}-r}}.
\end{split}
\end{align}
If $r=0$, then we cannot get a contribution from any of the four terms, since the shift for that term is larger in magnitude than the degrees of the sequences being correlated (cf.~\cref{Claudia}), and so $C_{n,s} =0$.
Likewise, if $0<r \leq\ell_{n-t}$, we cannot get a contribution from  $C_{n-t-1,r-3 \cdot\ell_{n-t-1}}$, while if $\ell_{n-t}\leq r<\ell_{n-t+1}$, we cannot get a contribution from  $\conj{C_{n-t-1,\ell_{n-t-1}-r}}$.
We prove the bounds for the three cases by dropping these non-contributory terms from \eqref{Mabel}, then taking absolute values of both sides, using the triangle inequality on the right-hand side, and bounding the absolute value of each specific correlation value $C_{k,u}$ with its largest possible magnitude $M_k$.
\end{proof}
The following technical result gives bounds on crosscorrelations for our Golay pairs $(x_n,y_n)$ in terms of the crosscorrelation and autocorrelation of the seed pair $(x_0,y_0)$, and is used in establishing the constant prefactor in the bound of \cref{Tommy}.
\begin{lemma}\label{Derrel}
If $s \in \Z$, $n \in \Z_+$, and $q=\floor{s/\ell_2}$, then
\begin{align*}
|C_{x_n,y_n}(s) |
\leq \big(|A_{n-1,q}| + |B_{n-1,q} | & +|\Gamma_{n-1,q} | + |\Delta_{n-1,q}|\big) \cdot  \PCC(x_0,y_0) \\
& + \big(|A_{n-1,q} | + |B_{n-1,q} |\big) \cdot 2 \PSL(x_0).
\end{align*}
\end{lemma}
\begin{proof}
Apply \cref{Nellie} with $t=n-1$ (noting that $q$ as defined here is the same $q$ as in the lemma) to see that
\[
|C_{n,s}| \leq (|A_{n-1,q}|+|B_{n-1,q}|) M_1 + (|\Delta_{n-1,q}|+|\Gamma_{n-1,q}|) M_0,
\]
and then note that $M_0=\PCC(x_0,y_0)$ and \cref{Demetrius} says that $M_1=\PCC(x_1,y_1) \leq 2 \PSL(x_0)+\PCC(x_0,y_0)$.
\end{proof}
\subsection{Inductive bounds}\label{Hildegard}
Recall the algebraic number $\alpha_0$, which appears in our bounds in Theorems \ref{Gale}--\ref{Tommy}, and whose technical details are described in \cref{Ollie}.
We show that if we have established an exponential bound with base $\alpha_0$ for crosscorrelation of Golay--Rudin--Shapiro sequences appearing earlier in the recursion, then in most cases we can deduce such a bound for Golay--Rudin--Shapiro sequences appearing later in the recursion, but this does not amount to a full inductive proof because there are exceptional sets of shifts (called $E_t$ in the result below) for which the inductive chain does not continue.
\begin{proposition}\label{Velma}
Let $n \in \N$ and suppose that there is some positive real number $K$ such that for all $m < n$ we have $\PCC(x_m) \leq K \alpha_0^m$.  For each $t \in \{1,2,\ldots,10\}$, define $E_t$ by
\begin{enumerate}
\item $E_1=(-1,1) \cdot \ell_{n-1}$,
\item $E_2=(-2,2) \cdot \ell_{n-2}$,
\item $E_3=\big( (-4  ,-2 )\cup (1,3 )\big)\cdot \ell_{n-3}$,
\item $E_4=\big( [-6 ,-5 )\cup (2 , 3 ) \cup [4 , 6)\big) \cdot \ell_{n-4}$,
\item $E_5=\big ([-12 ,-10 ) \cup (5 ,6) \cup (9,11 )\big)\cdot \ell_{n-5}$,
\item $E_6=\big((-23 , -21 ) \cup (10,12) \cup (18 , 19 ) \cup (21, 22 ) )\cdot \ell_{n-6}$,
\item $E_7=\big( (-46, -45 ) \cup (-43 ,  -42 ) \cup (21,22) \cup (37 , 38 ) \cup (42 ,43 ) \big) \cdot \ell_{n-7}$,
\item $E_8=\big( (-91 , -90 ) \cup(-86 , -85 ) \cup (42,43)  \cup (74 , 75 )\big)\cdot \ell_{n-8}$,
\item $E_9=\big((-182 , -180 ) \cup (-171 , -170) \cup (149 , 150  )\big)\cdot \ell_{n-9}$, and
\item $E_{10}=(-342, -341) \cdot \ell_{n-10}$.
\end{enumerate}
Then for each $t \in \{1,\ldots,10\}$, if $n > t$ and $s \not\in E_t$, then $|C_{n,s}| \leq K \alpha_0^n$.
\end{proposition}
\begin{proof}
Set $E_0=(-\ell_n,\ell_n)$.  When we prove the $t=1$ case, we may confine ourselves to $s \in E_0$, since otherwise $C_{n,s}=0$ (because $\deg x_n,\deg y_n < \ell_n$ by \cref{Claudia}). For each $t \in \{1,\ldots,10\}$, define $F_t=E_{t-1}\smallsetminus E_t$.
Since we prove cases with $t=1,2,\ldots,10$ in ascending order, the $t$th instance needs only be shown for $s \in F_t$.
We work out what the sets $F_1,\ldots,F_{10}$ are:
\begin{enumerate}
\item $F_1=\big((-2,-1] \cup [1,2)\big) \cdot \ell_{n-1}$,
\item $F_2=\emptyset$,
\item $F_3= \big([-2 ,1] \cup [3,4)\big)\cdot \ell_{n-3}$,
\item $F_4=\big((-8 , -6)\cup [-5 , -4 ) \cup [3 , 4 )\big)\cdot \ell_{n-4}$,
\item $F_5=\big((4, 5 ] \cup [8 , 9 ] \cup [11 , 12)\big)\cdot \ell_{n-5}$,
\item $F_6=\big([-24, -23 ] \cup [-21 , -20 ) \cup [19, 21]\big)\cdot \ell_{n-6}$,
\item $F_7=\big([-45, -43] \cup (20,21] \cup [22,24) \cup  (36 , 37 ] \cup [43 , 44))\cdot \ell_{n-7}$,
\item $F_8=\big((-92, -91 ] \cup [-85 , -84)\cup [43,44)  \cup [75 , 76) \cup (84 , 86 )\big) \cdot \ell_{n-8}$,
\item $F_9=\big((-172, -171 ] \cup(84,86) \cup (148, 149 ]\big)\cdot \ell_{n-9}$, and
\item $F_{10}=\big((-364, -360) \cup [-341 , -340 ) \cup (298 , 300)\big) \cdot \ell_{n-10}$.
\end{enumerate}
Now we write $s \in F_t$ as $s=q \cdot \ell_{n-t+1}+r$ where $q,r \in \Z$ and $0 \leq r < \ell_{n-t+1}$.
For each value of $t$, we give all the possibilities for $q$ and $r$, arranged into subcases.
\begin{enumerate}
\item For $t=1$, we have $q=-1$ and $r \in (0,\ell_{n-1}]$; or $q=0$ and $r \in [\ell_{n-1},\ell_n)$.
\item For $t=2$, there is nothing to prove.
\item For $t=3$, we have $q=-1$ and $r \in [0, \ell_{n-3})$; or $q=1$ and $r \in [\ell_{n-3},\ell_{n-2})$.
\item For $t=4$, we have $q=-4$ and $r \in (0,\ell_{n-3})$; $q=-3$ and  $r \in [\ell_{n-4},\ell_{n-3})$; or $q=1$ and $r \in [\ell_{n-4},\ell_{n-3})$. 
\item For $t=5$, we have $q=2$ and $r \in (0,\ell_{n-5}]$; $q=4$ and $r \in [0,\ell_{n-5}]$; or $q=5$ and $r \in [\ell_{n-5},\ell_{n-4})$.
\item For $t=6$, we have $q=-12$ and $r \in [0,\ell_{n-6}]$; $q=-11$ and $r \in [\ell_{n-6},\ell_{n-5})$; $q=9$ and $r \in [\ell_{n-6}, \ell_{n-5})$; or $q=10$ and $r \in [0, \ell_{n-6}]$. 
\item For $t=7$, we have $q=-23$ and $r \in [\ell_{n-7}, \ell_{n-6})$; $q=-22$ and $r \in [0,\ell_{n-7}]$; $q=10$ and $r \in (0,\ell_{n-7}]$; $q=11$ and $r \in [0,\ell_{n-6})$; $q=18$ and $r \in (0,\ell_{n-7}]$; or $q=21$ and $r \in [\ell_{n-7},\ell_{n-6})$.
\item For $t=8$, we have $q=-46$ and $r \in (0,\ell_{n-8}]$; $q=-43$ and $r \in [\ell_{n-8},\ell_{n-7})$; $q=21$ and $r \in [\ell_{n-8},\ell_{n-7})$;  $q=37$ and $r \in [\ell_{n-8},\ell_{n-7})$; or $q=42$ and $r \in (0,\ell_{n-7})$.
\item For $t=9$, we have
$q=-86$ and $r \in (0, \ell_{n-9}]$; $q=42$ and $r \in (0,\ell_{n-8})$; or $q=74$ and $r \in (0,\ell_{n-9}]$.
\item For $t=10$, we have $q=-182$ and $r \in (0,\ell_{n-9})$; $q=-181$ and $r \in [0,\ell_{n-9})$; $q=-171$ and $r \in [\ell_{n-10},\ell_{n-9})$; or $q=149$ and $r \in (0,\ell_{n-9})$.
\end{enumerate}
Then for each value of $t$, and for each specified $q$ and range of values of $r$ that we have listed for that $t$, we use the appropriate case of \cref{Nellie} with values of $A_{t,q}$, $B_{t,q}$, $\Gamma_{t,q}$, and $\Delta_{t,q}$ from Table \ref{Charles} to give a bound on $|C_{n,s}|$ where $s=q \cdot \ell_{n}+r$.
The bounds we get for each value of $t$ (considering all of its subcases listed above) are as follows.
\begin{enumerate}
\item For $t=1$, our two subcases give $|C_{n,s}| \leq M_{n-1}$; or $|C_{n,s}| \leq M_{n-1}$.
\item For $t=2$, there is nothing to prove.
\item For $t=3$, our two subcases give $|C_{n,s}| \leq 3 M_{n-3} + 2 M_{n-4}$; or  $|C_{n,s}| \leq 3 M_{n-3}$.
\item For $t=4$, our three subcases give $|C_{n,s}| \leq M_{n-4}+10 M_{n-5}$; $|C_{n,s}|\leq 3 M_{n-4}$; or  $|C_{n,s}| \leq 7 M_{n-4}$.
\item For $t=5$, our three subcases give $|C_{n,s} |\leq 7 M_{n-5}$; $|C_{n,s}| \leq 9 M_{n-5}$; or  $|C_{n,s}| \leq 11 M_{n-5}$.
\item For $t=6$, our four subcases give $|C_{n,s} |\leq 7 M_{n-6}$; $|C_{n,s}| \leq 9 M_{n-6}$; $|C_{n,s}| \leq 15 M_{n-6}$; or  $|C_{n,s}| \leq 17M_{n-6}$.
\item For $t=7$, our six subcases give $|C_{n,s} |\leq 33 M_{n-7}$; $|C_{n,s}| \leq 31 M_{n-7}$; $|C_{n,s}| \leq 27 M_{n-7}$; $|C_{n,s}| \leq 21M_{n-7}+14 M_{n-8}$; $|C_{n,s}| \leq 21M_{n-7}$; or  $|C_{n,s}| \leq 31M_{n-7}$.
\item For $t=8$, our five subcases give $|C_{n,s} |\leq 33 M_{n-8}$; $|C_{n,s}| \leq 49 M_{n-8}$; $|C_{n,s}| \leq 29 M_{n-8}$; $|C_{n,s}| \leq 45M_{n-8}$; or  $|C_{n,s}| \leq 13M_{n-8}+62 M_{n-9}$.
\item For $t=9$, our three subcases give $|C_{n,s} |\leq 55M_{n-9}$; $|C_{n,s}| \leq 83 M_{n-9} + 2 M_{n-10}$; or  $|C_{n,s}| \leq 87 M_{n-9}$.
\item For $t=10$, our four subcases give $|C_{n,s} |\leq 109 M_{n-10} + 66 M_{n-11}$; $|C_{n,s}| \leq 99 M_{n-10} + 66M_{n-11}$; $|C_{n,s}| \leq 153 M_{n-10}$; or  $117 M_{n-10} + 6 M_{n-11}$.
\end{enumerate}
Now by our hypothesis that $M_j \leq K \alpha_0^j$ for all $j <n$, so for each of the cases we know that if $s \in F_t$, then we have the following.
\begin{enumerate}
\item If $t=1$, then $|C_{n,s}| \leq K \alpha_0^{n-1}$.
\item If $t=2$, then nothing needs to be proved.
\item If $t=3$, then $|C_{n,s}| \leq  3 K \alpha_0^{n-3} + 2 K \alpha_0^{n-4}$.
\item If $t=4$, then $|C_{n,s}| \leq \max\{K \alpha_0^{n-4} + 10 K \alpha_0^{n-5}, 7 K \alpha_0^{n-4}\}$.
\item If $t=5$, then $|C_{n,s}| \leq 11 K \alpha_0 ^{n-5}$.
\item If $t=6$, then $|C_{n,s}| \leq 17 K \alpha_0 ^{n-6}$.
\item If $t=7$, then $|C_{n,s}| \leq \max\{33 K \alpha_0^{n-7}, 21 K \alpha_0^{n-7}+14 K \alpha_0^{n-8}\} $.
\item If $t=8$, then $|C_{n,s}| \leq \max\{49 K \alpha_0^{n-8}, 13 K \alpha_0^{n-8}+ 62 K \alpha_0^{n-9}\} $.
\item If $t=9$, then $|C_{n,s}| \leq \max\{87 K \alpha_0^{n-9}, 83 K \alpha_0^{n-9}+2 K \alpha_0^{n-10}\} $.
\item If $t=10$, then $|C_{n,s}| \leq \max\{109 K \alpha_0^{n-10}+66 K \alpha_0^{n-11}$$, $$153K \alpha_0^{n-10}$, $117 K\alpha_0^{n-10}+6 K \alpha_0^{n-11}\}.$ 
\end{enumerate}
Then one must show that these upper bounds are less than or equal to $K\alpha_0^n$, which one does by dividing the above bound for the $t$th case by $K \alpha_0^{n-t-1}$ and showing that this quotient is less than $\alpha_0^{t+1}$.
For example, if $t=7$, we must show that $33 \alpha_0 \leq \alpha_0^8$ and $21 \alpha_0 +14  \leq  \alpha_0^8$.  \cref{Ollie} shows how to verify such inequalities using purely rational arithmetic.
\end{proof}

\section{Correlations for recursive shifts}\label{Shifts}

In the last section, we iterated the fundamental crosscorrelation recursion \eqref{Reginald} from \cref{Lamar} many times in an attempt to provide an inductive proof of Theorems \ref{Gale}--\ref{Tommy}, but we found that there were exceptional shifts for which the induction could not be continued.
This section provides another inductive argument that also fails to work for some shifts, but will be shown to work for the critical shifts for which the method of \cref{Iterating} fails.

Recall that throughout this paper $\ell_n$, $x_n$, $y_n$ are always as in \cref{Stan}: $\ell_0$ is a positive integer and $x_0$ and $y_0$ are nonzero polynomials in $\C[z]$ of degree less than $\ell_0$ such that $(x_0,y_0)$ is a Golay complementary pair with $C_{x_0,x_0}(0)=C_{y_0,y_0}(0)$, and for every $n \in \N$, we have $\ell_n=2^n \ell_0$ and $(x_n,y_n)$ is the $n$th Golay pair obtained from $\ell_0$ and the seed pair $(x_0,y_0)$ via the Golay--Rudin--Shapiro recursion: $x_n(z)=x_{n-1}(z)+z^{\ell_{n-1}} y_{n-1}(z)$ and $y_n(z)=x_{n-1}(z)-z^{\ell_{n-1}} y_{n-1}(z)$.

The specific technique of this section is to construct a sequence $s_0,s_1,\ldots$ of shifts that obeys a particular recursion (see \cref{Thor}), and then use generating functions to express $C_{x_n,y_n}(s_n)$ in terms of the values $C_{x_m,y_m}(s_m)$ and $C_{x_m,y_m}(-s_m)$ for some $m < n$ (see \cref{Ares}).
This provides a bound (see \cref{Mars}) that will be used in \cref{Final Bounds} to fill in the gap left by the bounds in \cref{Iterating} (see the discussion before \cref{Nancy}).

\subsection{A recursion on shifts}\label{Thor}
We begin by defining a recursion for a sequence $s_0,s_1,\ldots$ of shifts, and then we explore the properties of sequences obeying this recursion.
\begin{definition}[Shift Recursion Rule]
We say that an infinite sequence of integers $s_0,s_1, \ldots $ follows the {\it shift recursion rule} to mean that $s_{n+1}=-s_n-\ell_n$ for each $n \in \N$.
\end{definition}
\begin{definition}[Standard Shift Sequence]\label{Kirk}
For each $n \in \N$, define $t_n = ((-1)^n-2^n)\ell_{0}/{3}.$ We call  the sequence $t_0,t_1, \ldots $ the {\it standard shift sequence}.
\end{definition}
\begin{remark}\label{Bethany}
Observe that the standard shift sequence $t_0,t_1, \ldots$ follows the shift recursion rule, that is, $ t_{n+1} = -t_n- \ell_{n}$.
\end{remark}
\begin{lemma}\label{Eric}
Suppose that $s_0, s_1,\ldots $ follows the shift recursion rule and let $t_0, t_1, \ldots$ be the standard shift sequence.  Then
$s_n=t_n + (-1)^n s_0$.
\end{lemma}
\begin{proof}
We induct on $n \in \N$ with the base case being  trivial since $t_0=0$.
Then we suppose that the statement holds for some $n\geq0$, add $\ell_n$ to both sides, and negate our equation to get $-s_n-\ell_n=-t_n-\ell_n + (-1)^{n+1} s_0$ . Then use the shift recursion rule to see that $s_{n+1} = t_{n+1} + (-1)^{n+1}s_0$.
\end{proof}
\begin{lemma}\label{Lody}
Suppose $s_0,s_1,\ldots$ follows the shift recursion rule.
Then there is some $n \in \N$ such that $|s_n| < \ell_n$.
Let $m$ be the least nonnegative integer such that $|s_m|<\ell_m$.
If $|s_0| < \ell_0$, then $m=0$.
If $s_0 \geq \ell_0$ (resp., $s_0 \leq -\ell_0$) then $m$ is the least even (resp., odd) integer with $m > -2 + \log_2\left({(-1)^m(3 s_0 + \ell_{0})/}{\ell_{0}}\right)$.
\end{lemma}
\begin{proof}
Let $t_0,t_1,\ldots$ be the standard shift sequence. 
By \cref{Eric} we have $s_n   = t_n + (-1)^n s_0  = {((-1)^n-2^n)\ell_{0}}/{3}  + (-1)^n s_0 = {(-1)^n(3 s_0+ \ell_{0})/}{3} - {\ell_n/}{3}$, so that
\begin{equation}\label{Shaun}
-\ell_{n} < s_n < \ell_{n} \text{ if and only if} -\frac{2}{3} \ell_{n} < \frac{(-1)^n (3 s_0 + \ell_{0})}{3} < \frac{4}{3} \ell_{n}.
\end{equation}
Since $\lim\limits_{n \to \infty} \ell_{n}= \lim\limits_{n \to \infty} 2^n \ell_{0} = \infty$, there is some $n \in \N$ such that $- \ell_{n} < s_n < \ell_{n}$, thus proving existence.
So we may indeed define $m$ to be the smallest nonnegative value of $n$ such that the equivalent statements in \eqref{Shaun} hold.
If $|s_0| < \ell_0$, then it is immediate that $m=0$ by the definition of $m$.

Suppose, for a contradiction, that $s_0 \geq \ell_{0}$ and $m$ is odd, or that $s_0 \leq - \ell_{0}$ and $m$ is even.
Then $-2 \ell_{m-1}/3 < 0 < (-1)^{m-1}(3s_0 + \ell_{0})/3$, and by negating the third inequality in \eqref{Shaun} with $n=m$ we have $(-1)^{m-1}(3s_0 + \ell_{0})/3 <  2 \ell_m /3 = 4 \ell_{m-1}/3$.
Therefore  $-2 \ell_{m-1}/3 < (-1)^{m-1}(3s_0 + \ell_{0})/3 < 4 \ell_{m-1}/3$, which by \eqref{Shaun} contradicts the minimality assumption on $m$.
Thus, $m$ must be even (resp., odd) if $s_0 \geq \ell_0$ (resp., $s_0 \leq -\ell_0$).
When $s_0 \geq \ell_0$ (resp., $s_0 \leq -\ell_0$) and $n$ is even (resp., odd), the third inequality in \eqref{Shaun} always holds, so $m$ must be the least even (resp., odd) value of $n$ that makes the fourth inequality in \eqref{Shaun} hold, i.e., the least even (resp., odd) integer such that $(-1)^m (3 s_0 + \ell_{0})/3 < 2^{m+2} \ell_0/3$.
\end{proof}
\begin{lemma}\label{Botswana}
If $ s_0, s_1,\ldots $ follows the shift recursion rule, then the following hold.
\begin{enumerate}[(i)]
\item\label{Eustace} If $0 \leq s_n < \ell_n$, then $-\ell_{n+1} < s_{n+1} \leq -\ell_n$.
\item\label{Ophelia} If $-\ell_n < s_n <0$, then $-\ell_n<s_{n+1}<0$.
\item\label{Jimmy} If $s_n \leq -\ell_n$, then $s_{n+1} \geq 0$.
\item\label{Eugene} If $s_n \geq \ell_n$, then $s_{n+1} \leq -\ell_{n+1}$.
\end{enumerate}
\end{lemma}
\begin{proof}
All four statements follow easily from the shift recursion rule. 
\end{proof}
\begin{corollary}\label{Nelson}
 Let $s_0,s_1,\ldots$ follow the shift recursion rule, and let $m$ be the smallest nonnegative integer such that $|s_m| < \ell_m$ (which exists by \cref{Lody}). Then we have the following:
\begin{enumerate}[(i)]
\item \label{Greg} $0 \leq s_m < \ell_m$ if $m>0$ or $\ell_0=1$,
\item \label{Jessica} $-\ell_n<s_n<0$ for all $n>m$, and 
\item For all $n \in \N$ with $n<m$, we have $s_n \geq \ell_n$ if $n \equiv m \pmod{2}$ and $s_n \leq -\ell_n$ if $n \not\equiv m \pmod{2}$. 
\end{enumerate}
\end{corollary}
\begin{proof}
If $m=0$ and $\ell_0=1$, then our assumption that $|s_m| < \ell_m$ forces $s_m=s_0=0$, and so $0 \leq s_m < \ell_m$.
If $m>0$, then $s_m>-\ell_m$ by assumption, so the contrapositive of \cref{Botswana}\eqref{Eugene} makes $s_{m-1}<\ell_{m-1}$, but we must have $|s_{m-1}| \geq \ell_{m-1}$ by the minimality of $m$, so that $s_{m-1}\leq -\ell_{m-1}$, which in turn implies $s_m \geq 0$ by \cref{Botswana}\eqref{Jimmy}, and thus our initial assumption that $|s_m| < \ell_m$ implies $0 \leq s_m < \ell_m$.
This proves \eqref{Greg}.

Now we prove \eqref{Jessica}.
Once we have an $n$ such that $-\ell_n<s_n<0$, then $- \ell_{n+1}<-\ell_n < s_{n+1}<0$ by \cref{Botswana}\eqref{Ophelia}. So if $-\ell_m<s_m<0$, we are done. Otherwise, $0 \leq s_m < \ell_m$, and then \cref{Botswana}\eqref{Eustace} shows that $-\ell_{m+1} < s_{m+1} \leq -\ell_m<0$.

For the last statement, since $|s_n| \geq \ell_n$ for all $n \in \N$ with $n < m$, \cref{Botswana}\eqref{Jimmy} and \eqref{Eugene} show that $s_0,s_1,\ldots,s_{m-1}$ is a sequence of nonzero integers with alternating signs, and since $s_m > -\ell_m$, its immediate predecessor, $s_{m-1}$, must be less than or equal to $-\ell_{m-1}$.
\end{proof}

\subsection{A generating function for crosscorrelations}\label{Ares}
We now prove some relations on the crosscorrelation of Golay--Rudin--Shapiro sequences evaluated at shifts satisfying the shift recursion rule. 
\begin{lemma}\label{Nestor}
If $s_0, s_1,\ldots $ follows the shift recursion rule, and $n > 1$ with $s_n < 0$, then
\begin{align*}
C_{x_n,y_n}(s_n) & = -C_{x_{n-1},y_{n-1}}(-s_{n-1}) + 2 C_{x_{n-2},y_{n-2}}(s_{n-2}), \text{and} \\
C_{x_n,y_n}(-s_n) & = \conj{C_{x_{n-1},y_{n-1}}(-s_{n-1}) }+2\conj{C_{x_{n-2},y_{n-2}}(s_{n-2})}.
\end{align*}
\end{lemma}
\begin{proof}
Use \cref{Geoff} with $s=s_n<0$ (resp., $s=-s_n>0$) and rewrite the shifts using the shift recursion rule to obtain the first (resp., second) relation.
\end{proof}
\begin{proposition}\label{Cecilia}
If $s_0,s_1, \ldots$ follows the shift recursion rule, $n > 2$, $s_{n-1} < 0$, and $s_n < 0$, then we have
\[
C_{x_n,y_n}(s_n) = \conj{C_{x_{n-1},y_{n-1}}(s_{n-1}) }+ 2 C_{x_{n-2},y_{n-2}}(s_{n-2}) -4 \conj{C_{x_{n-3},y_{n-3}}(s_{n-3})}.
\]
\end{proposition}
\begin{proof}
Consider the two relations in \cref{Nestor} when $n$ is replaced by $n-1$; if we add the conjugate of the first to the second and rearrange, we obtain
\[
C_{x_{n-1},y_{n-1}}(-s_{n-1}) = -\conj{C_{x_{n-1},y_{n-1}}(s_{n-1})} + 4\conj{C_{x_{n-3},y_{n-3}}(s_{n-3})},
\]
and if we write the expression on the right-hand side of this equation in place of the $C_{x_{n-1},y_{n-1}}(-s_{n-1})$ that appears in the first relation of \cref{Nestor}, we obtain the desired result.
\end{proof}
Now we use the above relations to obtain a generating function that organizes crosscorrelations of Golay--Rudin--Shapiro sequences at shifts following the shift recursion rule.
The generating function involves the algebraic number $\alpha_0=1.658967\ldots$, which is the unique real root of $X^3+X^2-2 X-4$, along with the two non-real conjugate roots, $\alpha_1$ and $\alpha_2$, all of which are discussed in \cref{Ollie}.
\begin{proposition}\label{Lily}
Let $s_0,s_1,\ldots$ be a sequence of integers following the shift recursion rule, and let $m$ be the least nonnegative integer such that $|s_m| < \ell_m$.
Let $\sigma \colon \C \to \C$ denote the complex conjugation automorphism, so $\sigma^k(x)=x$ if $k$ is even and $\sigma^k(x)=\conj{x}$ if $k$ is odd.
For each $v \in \{0,1\}$ and $k \in \N$, let $f_{v,k}=\sigma^{v+k}(C_{x_k,y_k}((-1)^v s_k))$.
Then for every $n \geq m$, we have
\[
f_{0,n} =\sum_{j \in \Z/3\Z} \sum_{v \in \{0,1\}}  E_{j,v} f_{v,m} (-\alpha_j)^{n-m},
\]
where 
$E_{j,0} = (2+ \alpha_{j+1} \alpha_{j+2})/ ((\alpha_j- \alpha_{j+1}) (\alpha_j - \alpha_{j+2}))$ and $E_{j,1} = -(1+ \alpha_{j+1} + \alpha_{j+2})/((\alpha_j- \alpha_{j+1}) (\alpha_j - \alpha_{j+2}))$ for each $j \in \Z/3\Z$.
\end{proposition}
\begin{proof}
Since $|s_k| \geq \ell_k$ for $k \in \{0,1,\ldots, m-1\}$, we know that $f_{0,k}=0$ for $k < m$.
Then for all $k \geq m$ we have $f_{0,k+3}-f_{0,k+2}-2f_{0,k+1}+4f_{0,k}=0$ by \cref{Cecilia} since $s_{k+3}$ and $s_{k+2}$ are negative by \cref{Nelson}\eqref{Jessica}.
This recursion has characteristic polynomial $X^3-X^2-2 X+4$ with roots $-\alpha_0$, $-\alpha_1$, and $-\alpha_2$.
Therefore,
\begin{equation}\label{Lawrence}
f_{0,n}= c_0 (-\alpha_0)^n + c_1 (-\alpha_1)^n + c_2 (-\alpha_2)^n
\end{equation}
for some constants $c_1$, $c_2$ and $c_3$. So we can solve a linear system to express $c_0$, $c_1$, and $c_2$ in terms of $f_{0,m}$, $f_{0,m+1}$, and $f_{0,m+2}$: \[
c_j = \frac{\alpha_{j+1} \alpha_{j+2} f_{0,m} + (\alpha_{j+1}+\alpha_{j+2}) f_{0,m+1} +f_{0,m+2}}{(-\alpha_j) ^m(\alpha_j - \alpha_{j+1})(\alpha_j - \alpha_{j+2})}
\]
for each $j \in \Z/3\Z$.
Now \cref{Nestor} with $n = m+2$ gives $f_{0,m+2}= - f_{1, m+1} + 2 f_{0,m}$ and when $n = m+1$, \cref{Nestor} gives $f_{1,m+1} = f_{1,m}$ (since $f_{0,k}=0$ when $k < m$) so that $f_{0, m+2} = -f_{1,m} + 2 f_{0,m}$. Also, \cref{Nestor} with $n=m+1$ gives $f_{0,m+1} = - f_{1,m}$ so we get
\[
c_j =(-\alpha_j)^{-m}\left[\frac{(2+\alpha_{j+1} \alpha_{j+2}) f_{0,m}-(1+\alpha_{j+1}+\alpha_{j+2})f_{1,m}}{(\alpha_j- \alpha_{j+1}) (\alpha_j - \alpha_{j+2})}\right]
\]
for $j \in \Z / 3\Z$.  We obtain the final result by substituting these values of $c_j$ into \eqref{Lawrence}.
\end{proof}
\subsection{Bounds from the generating function}\label{Mars}
Now we show how to use \cref{Lily} to obtain bounds on the peak crosscorrelation of Golay--Rudin--Shapiro sequences from bounds for sequences appearing earlier in the recursion.
This method almost provides an inductive proof of our bounds in Theorems \ref{Gale}--\ref{Tommy}, but it works only for crosscorrelation values at shifts satisfying a particular condition (the $n \geq m+9$ hypothesis in \cref{Mark} below), and so the chain of induction cannot always be completed.
It is possible to prove results similar to \cref{Mark} with less stringent conditions, but were not able to find a version that works in full generality.
Therefore we deemed it impractical to attempt an inductive proof Theorems \ref{Gale}--\ref{Tommy} based solely this generating function method, but we did find that this method works in the critical cases where the method from \cref{Iterating} fails; this enables us to prove our theorems in \cref{Final Bounds}.
\begin{proposition}\label{Mark}
Let $s_0,s_1,\ldots$ be a sequence of integers following the shift recursion rule, and let $m$ be the least nonnegative integer such that $|s_m| < \ell_m$. 
Suppose that $n \geq m+9$ and that $K$ is a positive real number and that $|C_{x_j,y_j}(s)| \leq K \alpha_0^j$ for all $s \in \Z$ and all $0 \leq j < n$.
Then $|C_{x_n,y_n}(s_n)| \leq K \alpha_0^n$.
\end{proposition}
\begin{proof}
For each $j \in \Z/3\Z$, $v \in \{0,1\}$, and $k \in \N$, define $E_{j,v}$ and $f_{v,k}$ as in \cref{Lily}.
Define $F_{v,g}  =\sum_{j \in \Z/3\Z} E_{j,v} (-\alpha_j/\alpha_0)^g $ for each $v \in \{0,1\}$ and $g \in \N$. 
We change the order of summation in the expression from \cref{Lily} to obtain
\[
f_{0,n} = \sum_{v \in \{0,1\}} f_{v,m} F_{v,n-m} \alpha_0^{n-m}.
\]
Divide both sides by $K \alpha_0^n$ and take the absolute value to get 
\begin{align*}
\frac{|f_{0,n}|}{K\alpha_0^n} & = \left| \sum_{v \in \{0,1\}} \frac{f_{v,m}}{K \alpha_0^m} F_{v,n-m} \right| \\
& \leq \sum_{v \in \{0,1\}} |F_{v,n-m}|,
\end{align*}
where the inequality comes from the triangle inequality and the fact that for each $v \in \{0,1\}$, we have $|f_{v,m}|/(K\alpha_0^m)=|C_{x_m,y_m}((-1)^v s_m)|/(K\alpha_0^m)$, which lies in the real interval $[0,1]$ by assumption.
Therefore, if we can show $F_{0,n-m} + (-1)^u F_{1,n-m}$ lies in the real interval $[-1,1]$ for each $u \in \{0,1\}$, then we will be done.

Define $G_{j,u} = E_{j,0} + (-1)^u E_{j,1}$ for each $u \in \{0,1\}$ and $j \in \Z/3\Z$.
So it suffices to show that $\sum_{j \in \Z/3\Z} G_{j,u} (-\alpha_j/\alpha_0)^{n-m}$ lies in $[-1,1]$, or equivalently, that $\sum_{j \in \Z/3\Z} G_{j,u} (\alpha_j/\alpha_0)^{n-m}$ lies in $[-1,1]$, for each $u \in \{0,1\}$.

Since $\alpha_0$ is real and $\alpha_1$ and $\alpha_2$ are a conjugate pair of non-reals, we see that $E_{0,v}$ is real and $E_{1,v}$ and $E_{2,v}$ are a conjugate pair for each $v \in \{0,1\}$, so that $G_{0,u}$ is real and $G_{1,u}$ and $G_{2,u}$ are a conjugate pair for each $u \in \{0,1\}$.
So it suffices to show that $ G_{0,u} + 2\Re(\rho^{n-m} G_{1,u})$ lies in $[-1,1]$, where $\rho=\alpha_1/\alpha_0$, or equivalently, it suffices to show that $\Re(\rho ^{n-m} G_{1,u})$ lies in $[(-1-G_{0,u})/2, (1-G_{0,u})/2]$ for each $u \in \{0,1\}$.

We may now use the techniques of \cref{Ollie} to show that $0< G_{0,u} <1$ for $u \in \{0,1\}$.
Thus it will suffice to show that $\Re(\rho ^{n-m} G_{1,u})$ lies in the smaller interval $[(-1+G_{0,u})/2, (1-G_{0,u})/2]$ for each $u \in \{0,1\}$.
At the end of \cref{Ollie}, it was noted that $|\alpha_1| < |\alpha_0|$, so $|\rho|<1 $.
Since $n-m \geq 9$, it suffices to prove that $|\rho|^9 \leq {(1-G_{0,u})}/{(2 |G_{1,u}|)}$, which is true if and only if $|\rho|^{18}\leq {(1-G_{0,u})^2}/{(4 |G_{1,u}|^2)}$, i.e., $\alpha_1^9 \alpha_2^9/\alpha_0^{18}\leq {(1-G_{0,u})^2}/{(4 G_{1,u} G_{2,u})}$, a fact that can be checked for each $u \in \{0,1\}$ using the procedures laid out in \cref{Ollie}.
\end{proof}

\section{Bounds from shifts in the recursion}\label{Final Bounds}

We are now ready to combine the bound from \cref{Velma} with the bound from \cref{Mark} into an inductive proof that gives a bound on the peak crosscorrelation of sequence pairs produced by the Golay--Rudin--Shapiro recursion.
Then we deduce bounds on the peak sidelobe level from these.
In \cref{Penny}, we show that shifts for which we cannot get a good bound on their crosscorrelation using \cref{Velma} are shifts for which we can get a good bound with \cref{Mark}.
Then in \cref{Eleanor}, we prove bounds on correlation for the particular case of  the Rudin--Shapiro sequences, and finally in \cref{Marisol} we prove more general bounds that cover the full range of Golay pairs generated by the Golay--Rudin--Shapiro recursion (\cref{Stan}).

Recall that throughout this paper $\ell_n$, $x_n$, $y_n$ are always as in \cref{Stan}: $\ell_0$ is a positive integer and $x_0$ and $y_0$ are nonzero polynomials in $\C[z]$ of degree less than $\ell_0$ such that $(x_0,y_0)$ is a Golay complementary pair with $C_{x_0,x_0}(0)=C_{y_0,y_0}(0)$, and for every $n \in \N$, we have $\ell_n=2^n \ell_0$ and $(x_n,y_n)$ is the $n$th Golay pair obtained from $\ell_0$ and the seed pair $(x_0,y_0)$ via the Golay--Rudin--Shapiro recursion: $x_n(z)=x_{n-1}(z)+z^{\ell_{n-1}} y_{n-1}(z)$ and $y_n(z)=x_{n-1}(z)-z^{\ell_{n-1}} y_{n-1}(z)$.

\subsection{Complementarity of Propositions \ref{Velma} and \ref{Mark}}\label{Penny}
We now show that the small range of shifts that the $t=10$ case of \cref{Velma} cannot handle can be dealt with using \cref{Mark}.
\begin{lemma}\label{Phyllis}
Let $n,s \in \Z$ with $n \geq 11$, and let $s_0,s_1,\ldots$ be the unique sequence following the shift recursion rule with $s_n=s$. Suppose that $s_n \in (-342 \ell_{n-10}, -341 \ell_{n-10})$, and  let $m$ be the least nonnegative integer such that $|s_m| < \ell_m$. Then $n-m \geq 11$.
\end{lemma}
\begin{proof}
If $|s_0| < \ell_0$, then $m=0$ and $n-m=n-0 \geq11$ by assumption.
So let us assume $|s_0| \geq \ell_0$ henceforth.
Let $t_0, t_1, \ldots$ be the standard shift sequence as in \cref{Kirk}.
Recall from \cref{Eric} that $(-1)^n s_0 = s_n-t_n$. Since $s_n $ lies in $(-342 \ell_{n-10}, -341 \ell_{n-10})$ and $t_n = ((-1)^n-2^n)\ell_0/3$, we have 
\[
(-1)^ n s_0 =s_n-t_n \in \left(\frac{-2 \cdot 2^{n-10} \ell_{0} -(-1)^n \ell_{0}}{3}, \frac{2^{n-10}\ell_{0} -(-1)^n \ell_{0}}{3}\right).
\]
In the case where $s_0 \geq \ell_0$, \cref{Lody} says that $m$ is the least even integer with $m>-2 + \log _2\left({(3 s_0 + \ell_{0})}/{\ell_{0}}\right)$.
When $n$ is even (resp., odd), we have $s_0 \in \left[\ell_0, {(2^{n-10} \ell_0 -\ell_{0})}/{3}\right)$ (resp., $s_0 \in \left[\ell_0, ({2^{n-9} \ell_{0} -\ell_{0}})/{3}\right)$), so then $-2 + \log _2\left( {(3 s_0 + \ell_{0})}/{\ell_{0}}\right)$ lies in $[0,n-12)$ (resp., $[0,n-11)$), and since $m$ is the least even integer greater than this, we have $m \leq n-12$ (resp., $m \leq n-11$). 
	
In the case where $s_0 \leq -\ell_0$, \cref{Lody} says that $m$ is the least odd integer with $m>-2 + \log _2\left({-(3 s_0 + \ell_{0})}/{\ell_{0}}\right)$.
When $n$ is even (resp., odd), we have $s_0 \in \left({(-2^{n-9} \ell_{0} -\ell_{0})}/{3},- \ell_{0}\right]$ (resp., $s_0 \in  \left({(-2^{n-10} \ell_{0} -\ell_{0})}/{3},- \ell_{0}\right]$), so then $-2 + \log _2\left( {-(3 s_0 + \ell_{0})}/{\ell_{0}}\right)$ lies in $[-1,n-11)$ (resp., $[-1,n-12)$), and since $m$ is the least odd integer greater than this, we have $m \leq n-11$ (resp., $m \leq n-12$). 
\end{proof}

\subsection{Bounds on peak correlation for Rudin--Shapiro sequences}\label{Eleanor}

Recall the algebraic number $\alpha_0= 1.658967\ldots$, which appears in our bounds in Theorems \ref{Gale}--\ref{Tommy}, and whose technical details are described in \cref{Ollie}.
\cref{Phyllis} now enables us to use \cref{Velma} and \cref{Mark} to prove the upper bounds in Theorems \ref{Gale} and \ref{Gabby}.
\begin{theorem}\label{Destiny}
If $n \in \N$ and $x_n$ and $y_n$ are the $n$th Rudin--Shapiro sequence and its companion, then $\PCC(x_n,y_n) \leq 5 \alpha_0^{n-3}=(1.095107\ldots) \alpha_0 ^n$. This upper bound becomes an equality if $n=3$; otherwise it is a strict inequality.
We also have $\PSL(x_n)= \PSL(y_n)\leq 5 \alpha_0^{n-4}=(0.660113...)\alpha_0 ^n$.
This upper bound becomes an equality if $n=4$; otherwise it is a strict inequality.
\end{theorem}
\begin{proof}
We begin by proving the bound on $\PCC(x_n,y_n)$.
We write $C_n$ for $x_n\conj{y_n}$, and then $C_{n,s}$ is the coefficient of $z^s$ in $C_n$; this coefficient is the same as $C_{x_n,y_n}(s)$.
We let $K=5\alpha_0^{-3}$, and we need to show that $|C_{n,s}| \leq K \alpha_0^n$ for all $n \in \N$ and $s \in \Z$.
We proceed by induction on $n$.
We shall prove the $n=0$ and $n=1$ cases directly below.
When $n \geq 2$, we set $t=n-1$ and use \cref{Velma} (noting that $\ell_{n-t}=\ell_1=2\ell_0=2$ for our Rudin--Shapiro sequences) to establish the bound for most shifts, leaving only some exceptional shifts that need to be checked.
When $n > 0$, we also pass over even values of $s$ in silence, since $C_{n,s}=0$ for those values by \cref{Vladwick}.
In what follows, when we assert an inequality involving powers of $\alpha_0$, we have verified it using the methods of \cref{Ollie}.
\begin{enumerate}
\setcounter{enumi}{-1}
\item If $n=0$, then since $x_0$ and $y_0$ are of degree $0$, we need only consider $s=0$, for which, using \cref{Destiny  Table}, we find that $C_{0,0} = 1 \leq 5 \alpha_0^{-3} = K \alpha_0^0$.
\item If $n=1$, then since $x_1$ and $y_1$ are of degree $1$, we need only consider $s \in \{-1,1\}$ for which, using \cref{Destiny Table}, we find that $|C_{1,s}| = 1 \leq 5 \alpha_0^{-2}= K \alpha_0^1 $.
\item If $n=2$, then \cref{Velma} with $t=1$ shows that we need only consider $s \in \{-1,1\}$ for which, using \cref{Destiny Table}, we find that $|C_{2,s}| \leq 3 \leq 5 \alpha_0^{-1} = K \alpha_0^2$.
\item If $n=3$, then \cref{Velma} with $t=2$ shows that we need only consider $s \in \{-3,1,1,3\}$ for which, using \cref{Destiny Table}, we find that $|C_{3,s}| \leq 5=K \alpha_0^3$.
\item If $n=4$, then \cref{Velma} with $t=3$ shows that we need only consider $s \in \{-7,-5,3,5\}$ for which, using \cref{Destiny Table}, we find that $|C_{4,s}| \leq 7 \leq 5\alpha_0 = K \alpha_0^4$.
\item If $n=5$, then \cref{Velma} with $t=4$ shows that we need only consider $s \in \{-11,5,9,11\}$ for which, using \cref{Destiny Table}, we find that $|C_{5,s}| \leq 13 \leq 5\alpha_0^2 = K \alpha_0^5$.
\item If $n=6$, then \cref{Velma} with $t=5$ shows that we need only consider  $s \in \{-23,-21,11,19,21\}$ for which, using \cref{Destiny Table}, we find that $|C_{6,s}| \leq 15\leq 5\alpha_0^3= K \alpha_0^6$.
\item If $n=7$, then \cref{Velma} with $t=6$ shows that we need only consider $s \in \{-45,-43,21,23,37,43\}$ for which, using \cref{Destiny Table}, we find that $|C_{7,s}| \leq 33\leq 5\alpha_0^4 = K \alpha_0^7$.
\item If $n=8$, then \cref{Velma} with $t=7$ shows that we need only consider $s \in \{-91,-85,43,75,85\}$ for which, using \cref{Destiny Table}, we find that $|C_{8,s}| \leq 49\leq 5\alpha_0^5 = K \alpha_0^8$.
\item If $n=9$, then \cref{Velma} with $t=8$ shows that we need only consider $s \in \{-181,-171,85,149\}$ for which, using \cref{Destiny Table}, we find $|C_{9,s}| \leq 83\leq 5\alpha_0^6 =K \alpha_0^9$.
\item If $n=10$, then \cref{Velma} with $t=9$ shows that we need only consider $s \in \{-363,-361,-341,299\}$ for which, using \cref{Destiny Table}, we find that $|C_{10,s}| \leq 153\leq 5\alpha_0^7 =K \alpha_0^{10}$.
\item Now suppose that $n \geq 11$ and that $|C_{j,s}|\leq K \alpha_0 ^j$ for all $0 \leq j<n$ and $s \in \Z$.
Then \cref{Velma} with $t=10$ shows us that we need only consider $s \in (-342 \ell_{n-10},-341 \ell_{n-10})$. Define $s_0,s_1,\ldots$ to be the unique sequence following the shift recursion rule with $s_n=s$.  Define $m$ to be the least nonnegative integer such that $|s_m| < \ell_m$. Then by \cref{Phyllis}, $n-m \geq 11 \geq 9$. Therefore \cref{Mark} shows $|C_{n,s} |\leq K \alpha_0 ^n$. 
\end{enumerate}
Thus, we have proved that $\PCC(x_n,y_n) \leq 5 \alpha_0 ^{n-3}$ for every $n\in\N$.
This inequality becomes an equality when $n=3$, because $|C_{x_3,y_3}(-3)|=|-5|=5 \alpha_0^{3-3}$ (see \cref{Destiny Table}).
When $n\not=3$, we know that $5\alpha_0^{n-3}$ is irrational by \ref{Ulysses}, and since $\PCC(x_n,y_n)$ must be an integer (since $x_n$ and $y_n$ are binary sequences), the upper bound must be a strict equality.

Since $x_0=y_0=1$, we have $C_{x_0,x_0}(s)=C_{y_0,y_0}(s)=0$ for all nonzero $s$, so $\PSL(x_0)=\PSL(y_0)=0 < 5\alpha_0^{0-4}$.
For $n > 0$ we have $\PSL(y_n) = \PSL(x_n) = \PCC(x_{n-1},y_{n-1}) \leq 5 \alpha_0 ^{(n-1)-3}$ (where the inequality becomes an equality if and only if $n-1=3$); the first equality comes from \cref{Sylvester}, the second equality is from \cref{Kenneth}, and the inequality is what we just proved.
Furthermore, we can show $5 \alpha_0^{n-3}=1.095107\ldots$ and $5 \alpha_0^{-4}=0.660113...$ using the methods of \cref{Ollie}.
\end{proof}
One might wonder if it is possible to devise an exponential bound for $\PCC(x_n,y_n)$ with a lower base than $\alpha_0$.
The following result shows that this is not possible.
\begin{proposition}\label{George}
If $x_n$ and $y_n$ the $n$th Rudin--Shapiro sequence and its companion for each $n \in \N$, and if $t_0,t_1,\ldots$ is the standard shift sequence, then
\[
\left|\frac{C_{x_n,y_n}(t_n)}{(-\alpha_0)^n} - \frac{9 \alpha_0^2 + 4 \alpha_0 + 6}{59}\right| \leq  2 \sqrt{\frac{-4\alpha_0^2 + 2\alpha_0 + 14}{59}} \left(\sqrt{\frac{\alpha_0^2-1}{2}}\right)^n,
\]
that is,
\[
\left|\frac{C_{x_n,y_n}(t_n)}{(-\alpha_0)^n} - 0.633990\ldots\right| \leq (0.654022\ldots)(0.935994\ldots)^n,
\]
and so,
\begin{align*}
\frac{\PCC(x_n,y_n)}{\alpha_0^n}
& \geq \frac{9 \alpha_0^2 + 4 \alpha_0 + 6}{59} - 2 \sqrt{\frac{-4\alpha_0^2 + 2\alpha_0 + 14}{59}} \left(\sqrt{\frac{\alpha_0^2-1}{2}}\right)^n \\
&  = (0.633990\ldots) - (0.654022\ldots)(0.935994\ldots)^n.
\end{align*}
We also have
\begin{align*}
\frac{\PSL(x_n)}{\alpha_0^n}
& \geq \frac{3\alpha_0^2+21\alpha_0+2}{118} - \sqrt{\frac{6\alpha_0^2 + 6\alpha_0 -16}{59}} \left(\sqrt{\frac{\alpha_0^2-1}{2}}\right)^n \\
&  = (0.382159\ldots) - (0.421193\ldots)(0.935994\ldots)^n.
\end{align*}
\end{proposition}
\begin{proof}
We first prove the bound involving $C_{x_n,y_n}(t_n)$.
Recall that the standard shift sequence $t_0,t_1,\ldots$ follows the shift recursion rule by \cref{Bethany}, and it has $-\ell_0 < t_0=0 < \ell_0$.
So we can use $t_0,t_1,\ldots$ as $s_0,s_1,\ldots$ with $m=0$ in \cref{Lily}, whose notations for $\sigma$, $f_{v,k}$ and $E_{j,v}$ we also adopt.
Since $t_0=0$ and $m=0$, we have $f_{0,m}=C_{x_0,y_0}(0)=1$ and $f_{1,m}=\conj{C_{x_0,y_0}(-0)}=1$.
For each $j \in \Z / 3 \Z$, set $E_j=E_{j,0}+E_{j,1}$.
Then, \cref{Lily} tells us that $\sigma^n(C_{x_n,y_n}(t_n))=E_0 (-\alpha_0)^n + E_1 (-\alpha_1)^n + E_2 (-\alpha_2)^n$, and since $x_n$ and $y_n$ are the Rudin--Shapiro sequences, we know that $C_{x_n,y_n}(t_n) \in \Z$, so in fact $C_{x_n,y_n}(t_n)=E_0 (-\alpha_0)^n + E_1 (-\alpha_1)^n + E_2 (-\alpha_2)^n$.
If we divide both sides by $-\alpha_0^n$, subtract $E_0$, and take the absolute value (recalling that, we obtain
\[
\left|\frac{C_{x_n,y_n}(t_n)}{(-\alpha_0)^n} - E_0\right|=\left|E_1\left(\frac{\alpha_1}{\alpha_0}\right)^n +E_2\left(\frac{\alpha_2}{\alpha_0}\right)^n\right|.
\]
Since $\alpha_0$ is real while $\alpha_1$ and $\alpha_2$ are complex conjugates, one can deduce from the definition of $E_{j,v}$ in \cref{Lily} that $E_{1,v}$ and $E_{2,v}$ are a conjugate pair for each $v \in \{0,1\}$, and so $E_1$ and $E_2$ are complex conjugates.
Thus
\begin{equation}\label{Leslie}
\left|\frac{C_{x_n,y_n}(t_n)}{(-\alpha_0)^n} -E_0\right| \leq 2 |E_1| \left|\frac{\alpha_1}{\alpha_0}\right|^n.
\end{equation}
The methods of \cref{Ollie} show that $E_0=E_{0,0}+ E_{0,1}$ is a positive number equal to $(9 \alpha_0^2 + 4 \alpha_0 + 6)/59 = 0.633990\ldots$.
Since $E_1$ and $E_2$ are conjugates, we have $|E_1|^2 = E_1 E_2$; the methods of \cref{Ollie} show this to be equal to $(-4\alpha_0^2 + 2\alpha_0 + 14)/59$, and they also show that $0.654022^2 < 4 E_1 E_2 < 0.654023^2$, so that $2 |E_1|=\sqrt{4 E_1 E_2}=0.654022\ldots$.
\cref{Gilda} shows that $|\alpha_1/\alpha_0|=\sqrt{(\alpha_0^2-1)/2}=0.935994\ldots$.

By definition, $\PCC(x_n,y_n) \geq |C_{x_n,y_n}(t_n)|$, so we have
\begin{align}
\begin{split}\label{Clarence}
\frac{\PCC(x_n,y_n)}{\alpha_0^n}
& \geq \left|\frac{C_{x_n,y_n}(t_n)}{(-\alpha_0)^n}\right| \\
& \geq E_0 - 2 |E_1| \left|\frac{\alpha_1}{\alpha_0}\right|^n,
\end{split}
\end{align}
where the second inequality uses \eqref{Leslie}, and numerical approximations of the quantities in the last expression have already been discussed above.

We now prove the bound on $\PSL(x_n)$.
Suppose that $n>0$.
Then \cref{Kenneth} shows that $\PSL(x_n)=\PCC(x_{n-1},y_{n-1})$, and we can use the bound we already proved to show that
\begin{align*}
\frac{\PSL(x_n)}{\alpha_0^n}
& = \frac{1}{\alpha_0} \cdot \frac{\PCC(x_{n-1},y_{n-1})}{\alpha_0^{n-1}} \\
& \geq \frac{1}{\alpha_0} \left(E_0 - 2 |E_1| \left|\frac{\alpha_1}{\alpha_0}\right|^{n-1}\right) \\
& =\frac{E_0}{\alpha_0} - 2 \left|\frac{E_1}{\alpha_1}\right| \left|\frac{\alpha_1}{\alpha_0}\right|^n,
\end{align*}
where we use \eqref{Clarence} to obtain the inequality.
We already have an expression for $|\alpha_1/\alpha_0|$.
We can use the methods of \cref{Ollie} to show that $E_0/\alpha_0$ is a positive number equal to $(3\alpha_0^2+21\alpha_0+2)/118=0.382159\ldots$.
Since $\alpha_0$ is positive real, while $\alpha_1$ and $\alpha_2$ are complex conjugates, and $E_1$ and $E_2$ are complex conjugates, we know that $(2|E_1/\alpha_1|)^2=4 E_1 E_2/(\alpha_1\alpha_2)$; the methods of \cref{Ollie} show this to be equal to $(6\alpha_0^2 + 6\alpha_0 -16)/59$, and they also show that $0.421193^2 < 4 E_1 E_2/(\alpha_1\alpha_2) < 0.421194^2$, so that $2 |E_1/\alpha_1|=\sqrt{4 E_1 E_2/(\alpha_1\alpha_2)}=0.431193\ldots$.
This proves the desired bound on $\PSL(x_n)$ for $n>0$, and the desired bound is trivial for $n=0$, since it is asserting that $\PSL(x_0)$ is larger than a negative number.
\end{proof}
\begin{remark}\label{Sergio}
\cref{Destiny} and \cref{George}, when taken together, show that $0.633990\ldots  \leq \limsup_{n \to \infty} \PCC(x_n,y_n)/\alpha_0^n \leq 1.095107\ldots$ when $x_n$ and $y_n$ are the $n$th Rudin--Shapiro sequence and its companion.
Therefore, our bound in \cref{Destiny} is sharp in the following sense:
\begin{itemize}
\item No exponential bound with a base $\beta < \alpha_0$ is possible.
\item If the bound is to hold for all $n \in \N$, we cannot lower the constant prefactor $1.095107\ldots$ since when $n=3$, the upper bound is met.
\end{itemize}
Similarly, we have $0.382159\ldots  \leq \limsup_{n \to \infty} \PSL(x_n)/\alpha_0^n \leq 0.660113\ldots$, and our bound on $\PSL(x_n)$ (and on $\PSL(y_n)$ since $\PSL(x_n)=\PSL(y_n)$ by \cref{Sylvester}) is sharp in the same sense as the bound for $\PCC(x_n,y_n)$ is.
\end{remark}
\cref{George} gives a lower bound on $\PCC(x_n,y_n)$ which is not very strong when $n$ is small (indeed, it is negative when $n=0$).
We can use our computer calculations of $\PCC(x_n,y_n)$ for small $n$ to prove a lower bound that is much better for small $n$, although for large $n$ the bound in \cref{George} is better.
\begin{proposition}\label{Vanessa}
If $x_n$ and $y_n$ are the $n$th Rudin--Shapiro sequence and its companion for each $n \in \N$, then we have $\PCC(x_n,y_n) \geq 133991557 \alpha_0^{n-38} = (0.593256\ldots) \alpha_0^n$, where the $\geq$ becomes an equality when $n=38$, but is a strict inequality otherwise.
This lower bound on $\PCC(x_n,y_n)$ is strictly stronger than that of \cref{George} for $n\leq 41$, but strictly weaker for $n\geq 42$.
For $n\geq 1$, we have $\PSL(x_n)=\PSL(y_n)\geq 133991557 \alpha_0^{n-39}=(0.357605\ldots) \alpha_0^n$, where the $\geq$ becomes an equality when $n=39$, but is a strict inequality otherwise.
This lower bound on $\PSL(x_n)$ is strictly stronger than that of \cref{George} for $n\leq 42$, but strictly weaker for $n\geq 43$.
\end{proposition}
\begin{proof}
Let $A=(9 \alpha_0^2 + 4 \alpha_0 + 6)/59$, $B=4(-4\alpha_0^2 + 2\alpha_0 + 14)/59$, and $C=(\alpha_0^2-1)/2$, so that the lower bound in \cref{George} says that $\PCC(x_n,y_n) \geq (A-\sqrt{B C^n}) \alpha_0^n$.
Let $D=133991557/\alpha_0^{38}$, which the methods of \cref{Ollie} show to be $0.593256\ldots$; the lower bound we wish to prove here is $\PCC(x_n,y_n) \geq D \alpha_0^n$.
Showing that this bound is strictly stronger (resp., weaker) than the bound of \cref{George} for a given $n$ is equivalent to showing that $A-D$ is less (resp., greater) than $\sqrt{B C^n}$, and since $A-D$ is positive (as can be verified using the methods of \cref{Ollie}), this is equivalent to showing that $(A-D)^2$ is less (resp., greater) than $B C^n$.
We know from \cref{George} that $C < 1$ (since $\sqrt{C}=0.935994\ldots$), and the methods of \cref{Ollie} can be used to show that $(A-D)^2 < B C^{41}$ but $(A-D)^2 > B C^{42}$, thus verifying the claim about the strength of the bound of \cref{George} relative to the bound $\PCC(x_n,y_n) \leq D \alpha_0^n$, which we wish to prove.
Thus we only need to prove the latter bound for $n \leq 41$, and from the values of $\PCC(x_n,y_n)$ we obtain by perusing \cref{Aeneas}, we use the methods of \cref{Ollie} to check that $\PCC(x_n,y_n) \leq 133991557 \alpha_0^{n-38}$ is indeed true for all such $n$.
When $n=38$, the value of $\PCC(x_n,y_n)$ from \cref{Aeneas} shows that this upper bound becomes an equality.
When $n\not=38$, we know that $133991557 \alpha_0^{n-38}$ is irrational by \ref{Ulysses}, and since $\PCC(x_n,y_n)$ must be an integer (since $x_n$ and $y_n$ are binary sequences), the lower bound must be a strict inequality.

For every $n > 0$ we have $\PSL(y_n) = \PSL(x_n) = \PCC(x_{n-1},y_{n-1}) \geq  133991557 \alpha_0^{(n-1)-38}$ (where the inequality becomes an equality if and only if $n-1=38$); the first equality comes from \cref{Sylvester}, the second equality is from \cref{Kenneth}, and the inequality is what we just proved.
The claim about the strength of this lower bound on $\PSL(x_n)$ relative to that of \cref{George} follows if one notes that both of these lower bounds are just the lower bounds for $\PCC(x_{n-1},y_{n-1})$ from here and from \cref{George}, and we know that former is stronger for $n-1\leq 41$ and the latter is stronger for $n-1\geq 42$.
The methods of \cref{Ollie} show that $133991557/\alpha_0^{39}$ is $0.357605\ldots$.
\end{proof}

\subsection{General bounds}\label{Marisol}
In the previous section, we focused specifically on Rudin--Shapiro sequences, but in this section, we consider Golay--Rudin--Shapiro sequences produced by \cref{Stan} in its full generality.
For this more general class of sequences, we obtain the correlation bounds given in \cref{Tommy}, which we restate here.
These general bounds are written in terms of the peak correlation values for the seed pair $(x_0,y_0)$, i.e., $\PCC(x_0,y_0)$ and $\PSL(x_0)$.
\begin{theorem}\label{Generic Destiny}
Suppose that $\ell_n$, $x_n$, and $y_n$ are as in \cref{Stan}, and that $K=9\alpha_0^{-4} \PCC(x_0,y_0) +18 \alpha_0^{-5} \PSL(x_0)$.
Then $\PCC(x_n,y_n)\leq K \alpha_0 ^n$ for every $n \in \N$ and $\PSL(x_n)= \PSL(y_n)\leq K \alpha_0 ^{n-1}$ for every $n > 0$.
\end{theorem}
\begin{proof}
We begin by proving the bound on $\PCC(x_n,y_n)$.
We write $C_n$ for $x_n\conj{y_n}$ and then $C_{n,s}$ is the coefficient of $z^s$ in $C_n$; this coefficient is the same as $C_{x_n,y_n}(s)$.
We let $L=\PCC(x_0,y_0)$ and $L'=\PSL(x_0)$. 
We proceed by induction on $n$.
We shall prove the $n=0$ and $n=1$ cases directly below.
When $n \geq 2$, we set $t=n-1$ and use \cref{Velma} to establish the bound for most shifts, leaving only some exceptional shifts that need to be checked.
For $2 \leq n \leq 10$, we check exceptional shift values $s$ using \cref{Derrel} with $q=\floor{s/\ell_2}$.
In what follows, when we assert an inequality involving powers of $\alpha_0$, we have verified it using the methods of \cref{Ollie}.
\begin{enumerate}
\setcounter{enumi}{-1}
\item If $n=0$, then $|C_{0,s}|\leq L \leq 9 \alpha_0^{-4} L+ 18 \alpha_0^{-5} L' = K \alpha_0^0$, where the first inequality follows from the definition of $L$, and the second follows from $1 \leq 9\alpha_0^{-4}$.
\item If $n=1$, then $|C_{1,s}|\leq L+2 L' \leq 9\alpha_0^{-3} L+18\alpha_0^{-4} L' =K \alpha_0^{1}$, where the first inequality follows by \cref{Demetrius}, and the second follows from $1 \leq 9\alpha_0^{-3}$ and $2 \leq 18\alpha_0^{-4}$.
\item If $n=2$, then \cref{Velma} with $t=1$  shows that we need only consider $s \in (- \ell_1,\ell_1)$, so that $q= \floor{s/{\ell_{2}}} \in \{-1,0\}$.
Then we have $|C_{2,s}| \leq 3 L + 1 \cdot 2 L' \leq 9\alpha_0^{-2} L+ 18 \alpha_0^{-3} L' =K \alpha_0^2$, where the first inequality follows by \cref{Derrel} using \cref{Charles} with $t=1$, and the second follows from $3 \leq 9\alpha_0^{-2}$ and $2 \leq 18 \alpha_0^{-3}$.
\item If $n=3$, then \cref{Velma} with $t=2$ shows that we need only consider $s \in (-2,2) \cdot \ell_1$, so that $q= \floor{s/{\ell_{2}}} \in \{-1,0\}$.
Then we have $|C_{3,s}| \leq 5 L + 3 \cdot 2 L' \leq 9\alpha_0^{-1} L+ 18\alpha_0^{-2} L' = K \alpha_0^3$, where the first inequality follows by \cref{Derrel} using \cref{Charles} with $t=2$, and the second follows from $5 \leq 9\alpha_0^{-1}$ and $6 \leq 18\alpha_0^{-2}$.
\item If $n=4$, then \cref{Velma} with $t=3$ shows that we need only consider $s \in \left((-4 ,-2 ) \cup (1,3 )\right)\cdot \ell_{1}$, so that $q=\floor{{s}/{\ell_{2}}} \in \{-2,0,1\}$.
Then we have $|C_{4,s}| \leq \max\{7 L + 5 \cdot 2 L', 9 L + 3 \cdot 2 L' \} \leq 9 L + 18\alpha_0^{-1} L' = K \alpha_0^4$, where the first inequality follows by \cref{Derrel} using \cref{Charles} with $t=3$, and the second follows from $10 \leq 18\alpha_0^{-1}$.
\item If $n=5$, then \cref{Velma} with $t=4$ shows that we need only consider $s \in ( [-6,-5) \cup (2,3) \cup [4,6))\cdot \ell_{1}$, so that $q = \floor{{s}/{\ell_{2}}} \in \{-3,1,2\}$.
Then we have $|C_{5,s}| \leq \max \{13 L+3 \cdot 2 L', 11 L + 9 \cdot 2 L' \} \leq  9 \alpha_0 L + 18 L'=K \alpha_0^5$, where the first inequality follows by \cref{Derrel} using \cref{Charles} with $t=4$, and the second follows from $13 \leq 9\alpha_0$.
\item If $n=6$, then \cref{Velma} with $t=5$  shows that we need only consider $s \in ([-12 ,-10 ) \cup (5,6) \cup (9 ,11)) \cdot \ell_1$, so that $q=\floor{{s}/{\ell_{2}}} \in \{-6,2,4,5\}$.
Then we have $|C_{6,s} | \leq \max \{15 L + 13\cdot 2 L',21 L + 7\cdot 2 L',17 L + 11\cdot 2 L'\} \leq 9 \alpha_0^2 L +18 \alpha_0 L'= K \alpha_0^6$, where the first inequality follows by \cref{Derrel} using \cref{Charles} with $t=5$, and the second follows from $21 \leq 9\alpha_0^2$ and $26 \leq 18\alpha_0$.
\item If $n=7$, then \cref{Velma} with $t=6$ shows that we need only consider $s \in\big((-23 , -21 ) \cup (10,12) \cup (18 , 19 ) \cup (21, 22 ) )\cdot \ell_{1} $, so that $q=\floor{{s}/{\ell_{2}}} \in \{-12,-11,5,9,10\}$.
Then we have $|C_{7,s}| \leq \max \{35 L+9\cdot 2 L', 27 L+21 \cdot 2 L', 33  L+15 \cdot 2 L', 31 L+17 \cdot 2 L' \} \leq 9 \alpha_0^3 L+18 \alpha_0^2 L' =K \alpha_0^7$, where the first inequality follows by \cref{Derrel} using \cref{Charles} with $t=6$, and the second inequality follows from $35 \leq 9\alpha_0^3$ and $42 \leq 18\alpha_0^2$.
\item If $n=8$, then \cref{Velma} with $t=7$  shows that we need only consider $s \in \big( (-46, -45 ) \cup (-43 ,  -42 ) \cup (21,22) \cup (37 , 38 ) \cup (42 ,43 ) \big) \cdot \ell_{1}$, so that $q=\floor{{s}/{\ell_{2}}} \in \{-23,-22,10,18,21\}$. 
Then we have  $|C_{8,s}| \leq \max \{47 L+33\cdot 2 L', 49 L+31 \cdot 2 L', 51 L+21 \cdot 2 L'\} \leq 9 \alpha_0^4 L+18 \alpha_0^3 L' =K \alpha_0^8$, where the first inequality follows by \cref{Derrel} using \cref{Charles} with $t=7$, and the second follows from $51 \leq 9\alpha_0^4$ and $66 \leq 18\alpha_0^3$.
\item If $n=9$, then \cref{Velma} with $t=8$ shows that we need only consider $s \in\big( (-91 , -90 ) \cup(-86 , -85 ) \cup (42,43)  \cup (74 , 75 ))\cdot \ell_{1}$, so that $q=\floor{{s}/{\ell_{2}}} \in \{-46,-43,21,37\}$.
Then we have  $|C_{9,s}| \leq \max\{99 L+33\cdot 2 L', 91 L+49 \cdot 2 L'\} \leq 9 \alpha_0^5 L+18 \alpha_0^4 L' = K \alpha_0^9$, where the first inequality follows by \cref{Derrel} using \cref{Charles} with $t=8$, and the second follows from $99 \leq 9\alpha_0^5$ and $98 \leq 18 \alpha_0^4$.
\item If $n=10$, then \cref{Velma} with $t=9$ shows that we need only consider  $s \in( (-182 , -180 ) \cup (-171 , -170 ) \cup (149 , 150 )) \cdot \ell_{1}$, so that $q=\floor{{s}/{\ell_{2}}} \in \{-91,-86,74\}$. 
Then we have $|C_{10,s}| \leq \max \{109 L+99\cdot 2 L', 153 L+55\cdot 2 L', 117 L+87 \cdot 2 L'\} \leq 9 \alpha_0^6 L+18 \alpha_0^5 L'=K \alpha_0^{10}$, where the first inequality follows by \cref{Derrel} using \cref{Charles} with $t=9$, and the second follows from $153 \leq 9\alpha_0^6$ and $198 \leq 18\alpha_0^5$.
\item Now suppose that $n \geq 11$ and that $|C_{j,s} | \leq K \alpha_0 ^j$ for all $0 \leq j<n$ and $s \in \Z$.
Then \cref{Velma} with $t=10$ shows that we need only consider $s \in (-342 \ell_{n-10},-341 \ell_{n-10})$.
Define $s_0,s_1,\ldots$ to be the unique sequence following the shift recursion rule that has $s_n=s$.
Define $m$ to be the least nonnegative integer  such that $|s_m| < \ell_m$.
Then by \cref{Phyllis}, $n-m \geq 11 \geq 9$.
Therefore, \cref{Mark} shows that $|C_{n,s} | \leq K \alpha_0 ^n$.
\end{enumerate}
Thus, we have proved that $\PCC(x_n,y_n) \leq K \alpha_0 ^n$ for every $n\in\N$.
For $n > 0$ we have $\PSL(y_n) = \PSL(x_n) = \PCC(x_{n-1},y_{n-1}) \leq K \alpha_0 ^{n-1}$, where the first equality comes from \cref{Sylvester}, the second is from \cref{Kenneth}, and the inequality is what we just proved.
\end{proof}
One might wonder if it is possible to devise an exponential bound for $\PCC(x_n,y_n)$ with a lower base than $\alpha_0$.
The following result shows that this is not possible.
\begin{proposition}\label{Generic George}
Let $\ell_n$, $x_n$, and $y_n$ be as in \cref{Stan}, and let $\sigma \colon \C\to\C$ be the complex conjugation automorphism.
Then there is some sequence $s_0,s_1,\ldots$ of integers following the shift recursion rule with $|s_0| < \ell_0$ such that
\[
M =\frac{(\alpha_0^2+7 \alpha_0 +40) C_{x_0,y_0}(s_0) + (17 \alpha_0^2 + \alpha_0-28) \conj{C_{x_0,y_0}(-s_0)}}{118}
\]
is nonzero, and if we set
\[
N =  |C_{x_0,y_0}(s_0)| \sqrt{\frac{-4 \alpha_0+24}{59}} +  |C_{x_0,y_0}(-s_0)| \sqrt{\frac{-2\alpha_0^2+10\alpha_0+12}{59}},
\]
then 
\begin{align*}
\left|\frac{\sigma^n(C_{x_n,y_n}(s_n))}{(-\alpha_0)^n}-M \right|
& \leq N  \left(\sqrt{\frac{\alpha_0^2-1}{2}}\right)^n \\
& = N (0.935994\ldots)^n,
\end{align*}
and so
\begin{align*}
\frac{\PCC(x_n,y_n)}{\alpha_0^n}
& \geq |M| - N  \left(\sqrt{\frac{\alpha_0^2-1}{2}}\right)^n \\
& = |M| - N (0.935994\ldots)^n.
\end{align*}
\end{proposition}
\begin{proof}
Since $\alpha_0$ is a real number and $\{1,\alpha_0,\alpha_0^2\}$ is a $\Q$-linearly independent set, it is clear that $-(17 \alpha_0 ^2 + \alpha_0 -28)/(\alpha_0 ^2 + 7 \alpha_0 +40)$ is a real number not equal to $1$ or $-1$.
We use \cref{Patty} to obtain some $s_0$ with $|s_0| < \ell_0$ such that $(\alpha_0^2+7 \alpha_0 +40) C_{x_0,y_0}(s_0) + (17 \alpha_0^2 + \alpha_0-28) \conj{C_{x_0,y_0}(-s_0) }\neq 0$, i.e., $M\not=0$.
Then choose $s_1,s_2 \ldots$ so that $s_0,s_1,s_2,\ldots$ follows the shift recursion rule.
Define $f_{v,k}$ and $E_{j,v}$ as in \cref{Lily}, and one can use the methods of \cref{Ollie} to compute that $E_{0,0}=(\alpha_0^2+7 \alpha_0 +40)/118$ and $E_{0,1}=(17 \alpha_0^2 + \alpha_0-28)/118$, so that $M=E_{0,0} f_{0,0} + E_{0,1} f_{1,0}$.
Then note that $m=0$ in \cref{Lily} (because $|s_0| < \ell_0$) so we obtain $f_{0,n}= \sum_{j \in \Z/3\Z} \sum_{v \in \{0,1\}}  (-\alpha_j)^n E_{j,v} f_{v,0}$, which we can rearrange and apply the definition of $f_{v,k}$ to obtain
\[
\frac{\sigma^n(C_{x_n,y_n}(s_n))}{(-\alpha_0)^n} - M = \sum_{j \in \{1,2\}} \sum_{v \in \{0,1\}} E_{j,v} \sigma^v(C_{x_0,y_0}((-1)^v s_0)) \left(\frac{\alpha_j}{\alpha_0}\right)^n.
\]
If we take the absolute value of both sides, apply the triangle inequality, remember that $\alpha_1$ and $\alpha_2$ are complex conjugates, and recognize that $E_{1,v}$ and $E_{2,v}$ is a conjugate pair for each $v \in \{0,1\}$ (see the definitions in \cref{Lily}), we obtain
\[
\left|\frac{\sigma^n(C_{x_n,y_n}(s_n))}{(-\alpha_0)^n} - M\right| \leq \left|\frac{\alpha_1}{\alpha_0}\right|^n \sum_{v \in \{0,1\}} 2 |E_{1,v}| |C_{x_0,y_0}((-1)^v s_0)|.
\]
Then we use the methods of \cref{Ollie} to see that $4 |E_{1,0}|^2 = 4 E_{1,0} E_{2,0} = (-4\alpha_0+24)/59$ and $4 |E_{1,1}|^2=4 E_{1,1} E_{2,1} = (-2\alpha_0^2+10\alpha_0+12)/59$, so that we recognize the sum over $v$ on the right-hand side as $N$, while \cref{Gilda} shows that $|\alpha_1/\alpha_0|=\sqrt{(\alpha_0^2-1)/2}=0.935994\ldots$.
This completes the proof of the first bound, and the second bound follows immediately since $\PCC(x_n,y_n)$ upper bounds all crosscorrelation values for $x_n$ with $y_n$.
\end{proof}
\begin{remark}\label{Generic Sergio}
\cref{Generic Destiny} and \cref{Generic George}, when taken together, show that $\limsup_{n \to \infty} \PCC(x_n,y_n)/\alpha_0^n$ is a strictly positive real number.
Therefore, no exponential bound on $\PCC(x_n,y_n)$ with a base $b < \alpha_0$ is possible.
Furthermore, by \cref{Sylvester} and \cref{Kenneth}, which say that $\PSL(y_n)=\PSL(x_n)=\PCC(x_{n-1},y_{n-1})$ for $n \geq 1$, we know that our bounds on $\PSL(x_n)$ and $\PSL(y_n)$ are also sharp in the same sense.
\end{remark}

\section*{Acknowledgement}
The authors thank an anonymous reviewer for suggestions that helped improve the paper.

\newcommand{\etalchar}[1]{$^{#1}$}

\end{document}